\documentclass[a4paper,english,numberwithinsect,cleveref,thm-restate]{lipics-v2021}
\usepackage{booktabs}
\usepackage{xspace}
\usepackage{microtype} 

\graphicspath{{./figures/}}

\newcommand{\etal}{et al.}

\newcommand{\R}{\ensuremath{\mathbb{R}}\xspace}

\newcommand{\mkmcal}[1]{\ensuremath{\mathcal{#1}}\xspace}
\newcommand{\W}{\mkmcal{W}}
\newcommand{\E}{\mkmcal{E}}
\newcommand{\A}{\mkmcal{A}}
\newcommand{\U}{\mkmcal{U}}
\renewcommand{\L}{\mkmcal{L}}
\newcommand{\level}{\textit{l}}

\DeclareMathOperator{\polylog}{polylog}
\DeclareMathOperator{\interior}{int}

\newcommand{\ol}[1]{\overline{#1}}

\newcommand{\pluseq}{\mathrel{+}=}

\newcommand{\North}{\ensuremath{\mathrm{North}}\xspace}
\newcommand{\East}{\ensuremath{\mathrm{East}}\xspace}
\newcommand{\South}{\ensuremath{\mathrm{South}}\xspace}
\newcommand{\West}{\ensuremath{\mathrm{West}}\xspace}

\newcommand{\splt}{\ensuremath{\mathrm{split}}}
\newcommand{\stab}{\ensuremath{\mathrm{stab}}}
\newcommand{\argmax}{\ensuremath{\mathrm{argmax}}}
\newcommand{\ply}{\ensuremath{\mathrm{ply}}}
\newcommand{\CH}{\ensuremath{\mathit{CH}}\xspace}

\title{Robust Bichromatic Classification using Two Lines}

\author{Erwin Glazenburg}{Utrecht University, The Netherlands}{e.p.glazenburg@uu.nl}{https://orcid.org/0009-0003-6645-4240}{}
\author{Thijs van der Horst}{Utrecht University, The Netherlands \and TU Eindhoven, The Netherlands}{t.w.j.vanderhorst@uu.nl}{https://orcid.org/0009-0002-6987-4489}{}
\author{Tom Peters}{TU Eindhoven, The Netherlands}{t.peters1@tue.nl}{https://orcid.org/0000-0002-2702-7532}{}
\author{Bettina Speckmann}{TU Eindhoven, The Netherlands}{b.speckmann@tue.nl}{https://orcid.org/0000-0002-8514-7858}{}
\author{Frank Staals}{Utrecht University, The Netherlands}{f.staals@uu.nl}{https://orcid.org/0009-0004-8522-1351}{}

\authorrunning{E. Glazenburg, T. van der Horst, T. Peters, B. Speckmann, and F. Staals}

\Copyright{E. Glazenburg, T. van der Horst, T. Peters, B. Speckmann, and F. Staal}



\hideLIPIcs
\nolinenumbers

\begin{document}

\maketitle

\begin{abstract}
  Given two sets $R$ and $B$ of $n$ points in the plane, we
  present efficient algorithms to find a two-line linear classifier
  that best separates the ``red'' points in $R$ from the ``blue''
  points in $B$ and is robust to outliers. More precisely, we find a
  region $\W_B$ bounded by two lines, so either a halfplane, strip,
  wedge, or double wedge, containing (most of) the blue points $B$,
  and few red points. Our running times vary between optimal
  $O(n\log n)$ and around $O(n^3)$, depending on the type of
  region $\W_B$ and whether we wish to minimize only red outliers,
  only blue outliers, or both.
\end{abstract}

\section{Introduction}
\label{sec:Introduction}

Let $R$ and $B$ be two sets of at most $n$ points in the plane. Our
goal is to best separate the ``red'' points $R$ from the ``blue''
points $B$ using at most two lines. That is, we wish to find a region
$\W_B$ bounded by lines $\ell_1$ and $\ell_2$ containing (most of) the
blue points $B$ so that the number of points $k_R$ from $R$ in the
interior $\interior(\W_B)$ of $\W_B$ and/or the number of points $k_B$
from $B$ in the interior of the region $\W_R = \R^2 \setminus \W_B$ is
minimized. We refer to these subsets $\E_R = R \cap \interior(\W_B)$
and $\E_B = B \cap \interior(\W_R)$ as the red and blue outliers,
respectively, and define $\E = \E_R \cup \E_B$ and $k=k_R+k_B$.

\begin{figure}[tb]
    \centering
    \includegraphics{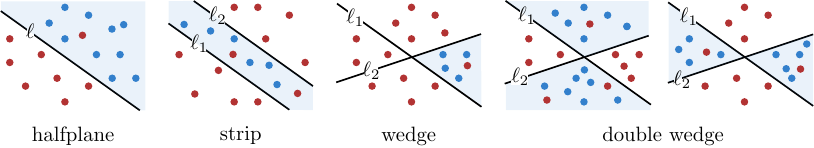}
    \caption{We consider separating $R$ and $B$ by at most two
      lines. This gives rise to four types of regions $\W_B$:
      halfplanes, strips, wedges, and two types of double wedges: hourglasses and bowties.
    }  
    \label{fig:allCases}
\end{figure}

The region $\W_B$ is either: $(i)$ a halfplane, $(ii)$ a
\emph{strip} bounded by two parallel lines $\ell_1$ and $\ell_2$,
$(iii)$ a \emph{wedge}, i.e., one of the four regions induced by a pair
of intersecting lines $\ell_1$ and $\ell_2$, or $(iv)$ a \emph{double wedge},
i.e., two opposing regions induced by a pair
of intersecting lines $\ell_1$ and $\ell_2$ (we further distinguish \emph{hourglass} double wedges, that contain a vertical line, and the remaining \emph{bowtie} double wedges). See
Figure~\ref{fig:allCases}. We can reduce the case
that $\W_B$ would consist of three regions to the single-wedge case by
recoloring the points. For each of these cases for the shape of $\W_B$ we consider three problems: allowing only red outliers ($k_B = 0$) and minimizing $k_R$, allowing only blue outliers ($k_R = 0$) and minimizing $k_B$, or allowing both outliers and minimizing $k$. We present efficient algorithms for each of these problems, as shown in Table~\ref{tab:overview_results}.

\subparagraph{Motivation and related work.} 
Classification is a key problem in computer
science. The input is a labeled set of points and the goal is
to obtain a procedure that, given an unlabeled point, assigns it a label
that ``fits it best'', considering the labeled points. Classification has many
direct applications, e.g. identifying SPAM in email messages,
or tagging fraudulent transactions~\cite{sculley07relax_onlin_svms_spam_filter,shen2007application},
but is also the key subroutine in other problems such as
clustering~\cite{DBLP:books/crc/aggarwal2014}. 

We restrict our attention to
binary classification where our input is a set $R$ of red points and a set $B$ of blue points.
We can compute whether $R$ and $B$ can be
perfectly separated by a line (and compute such a line if it exists)
in $O(n)$ time using linear programming. This extends to finding a
separating hyperplane in case of points in $\R^d$, for some constant
$d$~\cite{megiddo84linear_progr_linear_time_when}. 

\begin{figure}[b]
  \centering
  \includegraphics{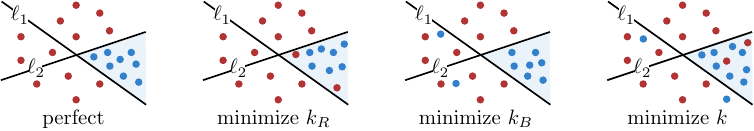}
  \caption{Perfectly separating $R$ and $B$ may require more than one
    line. When considering outliers, we may allow `(and minimize) only red outliers, only blue outliers, or both.
  }
  \label{fig:outlier_types}
\end{figure}

Clearly, it is not always possible to find a hyperplane that
perfectly separates the red and the blue points, see for example
Figure~\ref{fig:outlier_types}, in which the blue points are actually
all contained in a wedge. 
Hurtato
\etal~\cite{hurtado04separ,hurtado01separ} consider separating $R$ and
$B$ in $\R^2$ using at most two lines $\ell_1$ and $\ell_2$. In this
case, linear programming is unfortunately no longer
applicable. Instead, they present $O(n\log n)$ time
algorithms to compute a perfect separator (i.e., a strip, wedge, or
double wedge containing all blue points but no red points), if it
exists. These results were shown to be
optimal~\cite{arkin06some_lower_bound_geomet_separ_probl}, and can be
extended to the case where $B$ and $R$ contain other geometric objects
such as segments or circles, or to include constraints on the
slopes~\cite{hurtado04separ}. Similarly, Hurtado
\etal~\cite{hurtado05red} considered similar strip and wedge
separability problems for points in $\R^3$. Arkin
\etal~\cite{arkin12separ_point_sets_level_linear_class_trees} show how
to compute a 2-level binary space partition (a line $\ell$ and two
rays starting on $\ell$) separating $R$ and $B$ in $O(n^2)$ time, and
a minimum height $h$-level tree, with $h \leq \log n$, in
$n^{O(\log n)}$ time. Even today, computing perfect bichromatic
separators with particular geometric properties remains an active research topic~\cite{alegria23separ}.

Alternatively, one can consider separation with a (hyper-)plane but allow for outliers.
Chan~\cite{chan05low_dimen_linear_progr_violat} presented
algorithms for linear programming in $\R^2$ and $\R^3$ that allow for
up to $k$ violations --and thus solve hyperplane separation with up to
$k$ outliers-- that run in $O((n+k^2)\log n)$ and
$O(n\log n + k^{11/4}n^{1/4}\polylog n)$ time, respectively. In higher
(but constant) dimensions, only trivial solutions are known. For
arbitrary (non-constant) dimensions the problem is
NP-hard~\cite{amaldi95compl_approx_findin_maxim_feasib}. There is also a fair amount of work that
aims to find a halfplane that minimizes some other error measure, e.g. the distance to the farthest misclassified point, or the sum of
the distances to misclassified
points~\cite{aronov12minim,har-peled05separ_outlier}.

Separating points using more general non-hyperplane separators and
outliers while incorporating guarantees on the number of outliers
seems to be less well studied. Seara~\cite{seara2002geometric} showed
how to compute a strip containing all blue points, while minimizing
the number of red points in the strip in $O(n\log n)$ time. Similarly,
he presented an $O(n^2)$ time algorithm for computing a wedge with the
same properties. Armaselu and Daescu~\cite{armaselu19dynam} show how
to compute and maintain a smallest circle containing all red points
and the minimum number of blue points. In this paper, we take some further steps toward the fundamental, but challenging problem of computing a robust non-linear separator that provides performance guarantees. 

\newcommand{\our}[1]{\ensuremath{#1}}

\begin{table}[tb]
  \caption{An overview of our results. Expected running times are
    marked with a $\ddag$.}
  \centering
  \begin{tabular}{l llllll}
    \toprule
    region $\W_B$ & \multicolumn{2}{c}{minimize $k_R$} & \multicolumn{2}{c}{minimize $k_B$} & \multicolumn{2}{c}{minimize $k$} \\
    \midrule
    halfplane 
                 & $O(n\log n)$ &\S\ref{sec:Single_line_separation_one-sided}
                 & $O(n\log n)$ &\S\ref{sec:Single_line_separation_one-sided}
                 & $O((n+k^2)\log n)$&\cite{chan05low_dimen_linear_progr_violat}\\
    \midrule
    strip 
             & $O(n\log n)$&\cite{seara2002geometric}, \S\ref{sub:Strip_Red_outliers}
             & $O(n^2\log n)$   &\S\ref{sub:Strip_Blue_outliers}
             & $O(n^2\log n)$   &\S\ref{sub:strip_Both_outliers}\\
    \midrule
    wedge 
                 & $O(n^2)$    &\cite{seara2002geometric}
             &$O(n^{5/2} \log n) \ddag$& \S\ref{sub:Single_wedge_with_blue_outliers} 
             & & \\
             
                 & \our{O(n\log n)}
               &\S\ref{sub:Single_wedge_with_red_outliers}
             & $O(n k_B^2 \log^2 n \log k_B)$ &\S\ref{sub:single_wedge_both_outliers}
             & $O(nk^2 \log^3 n \log k)$ &\S\ref{sub:single_wedge_both_outliers} \\

    \midrule
    double bowtie
               & \our{O(n^2)} &\S\ref{sub:bowtie_red}
               & \our{O(n^2 \log n)} &\S\ref{sub:bowtie_blue}
               & $O(n^2 k \log^3 n \log k)$ & \S\ref{sub:Hourglass_wedge_separation_with_both_outliers}\\
    \bottomrule
  \end{tabular}
  \label{tab:overview_results}
\end{table}

\subparagraph{Results.} 
We present efficient algorithms for computing
a region $\W_B = \W_B(\ell_1,\ell_2)$ defined by at most two lines
$\ell_1$ and $\ell_2$ containing only the blue points, that are robust
to outliers. Our results depend on the type of region $\W_B$ we are
looking for, i.e., halfplane, strip, wedge, or double wedge, 
as well as on the type of outliers we allow: red outliers (counted by $k_R$), blue outliers (counted by $k_B$), or all outliers (counted by $k$).
Refer to Table~\ref{tab:overview_results} for an overview.

Our main contributions are efficient algorithms for when $\W_B$ is
really bounded by two lines. These versions can be solved by a simple
$O(n^4)$ time algorithm that explicitly considers all candidate
regions, refer to Section~\ref{sub:brute_force}. However, we show that
these versions can actually be solved significantly faster as well.

In particular in the versions where we minimize the number of red
outliers $k_R$ we achieve significant speedups. For example, we can
compute an optimal wedge $\W_B$ containing $B$ and minimizing $k_R$ in
optimal $\Theta(n\log n)$ time (which improves an earlier $O(n^2)$
time algorithm from Seara~\cite{seara2002geometric}). We use two types
of duality transformations that allow us to map each point
$p \in R \cup B$ into a \emph{forbidden region} $E_p$ in a
low-dimensional parameter space, such that: \textit{i)} every point
$s$ in this parameter space corresponds to a region $\W_B(s)$, and
\textit{ii)} this region $\W_B(s)$ misclassifies point $p$ if and only
if this point $s$ lies in $E_p$. This allows us to solve the problem
by computing a point that lies in the minimum number of forbidden
regions.

Surprisingly, the versions of the problem in which we minimize the
number of blue outliers $k_B$ are much more challenging. For none of
these versions we can match our running times for minimizing $k_R$,
while needing more advanced tools. For example, for the single wedge
version, we use dynamic lower envelopes to obtain a
batched query problem that we solve using spanning-trees with low
stabbing number~\cite{chazelle89quasi_optim_range_searc_space}. See
Section~\ref{sub:Single_wedge_with_blue_outliers}.

For the case where both red and blue outliers are allowed and we minimize $k$, we present output-sensitive algorithms whose running time depends on the optimal value of $k$. We essentially fix one of the lines $\ell_1$, and use linear
programming (LP) with
violations~\cite{chan05low_dimen_linear_progr_violat,
  glazenburg2024dynamicsvm} to compute an optimal line $\ell_2$ that
together with $\ell_1$ defines $\W_B$. We show that by using results
on $\leq k$-levels, a recent data structure for dynamic LP with
violations~\cite{glazenburg2024dynamicsvm}, and binary
searching, we can achieve algorithms with running times around
$O(n^2k\polylog n)$.


\subparagraph{Outline.}
We give some additional definitions and notation in Section~\ref{sec:Preliminaries}. In Section~\ref{sec:propertiesOptimalSeparator} we present a characterization of optimal solutions that lead to our simple $O(n^4)$ time algorithm for any type
of wedges, and in Section~\ref{sec:Single_line_separation_one-sided} we show how to extend Chan's algorithm~\cite{chan05low_dimen_linear_progr_violat} for linear programming with violations to handle one-sided outliers. In Sections \ref{sec:Separation_with_a_strip}, \ref{sec:single_wedge}, and \ref{sec:double_wedge} we discuss the case when $\W_B$ is, respectively, a strip, wedge, or double wedge. In each of these sections we separately go over minimizing the number of red outliers $k_R$, the number of blue outliers $k_B$, and the total number of outliers $k$. We wrap up with some concluding remarks and future work in
Section~\ref{sec:Concluding_Remarks}.

\section{Preliminaries}
\label{sec:Preliminaries}

In this section we discuss some notation and concepts used throughout
the paper.
For ease of exposition we assume $B \cup R$ contains at least three
points and is in general position, i.e., that all coordinate values are
unique, and that no three points are colinear.

\subparagraph{Notation.} Let $\ell^-$ and $\ell^+$ be the two
halfplanes bounded by line $\ell$, with $\ell^-$ below $\ell$ (or left
of $\ell$ if $\ell$ is vertical). Any pair of lines $\ell_1$ and
$\ell_2$, with the slope of $\ell_1$ smaller than that of $\ell_2$,
subdivides the plane into at most four interior-disjoint regions
$\North(\ell_1,\ell_2)= \ell_1^+ \cap \ell_2^+$,
$\East(\ell_1,\ell_2) = \ell_1^+ \cap \ell_2^-$,
$\South(\ell_1,\ell_2) = \ell_1^- \cap \ell_2^-$ and
$\West(\ell_1,\ell_2) = \ell_1^- \cap \ell_2^+$. When $\ell_1$ and
$\ell_2$ are clear from the context we may simply write \North to mean
$\North(\ell_1,\ell_2)$ etc. We assign each of these regions to either
$B$ or $R$, so that $\W_B=\W_B(\ell_1,\ell_2)$ and
$\W_R=\W_R(\ell_1,\ell_2)$ are the union of some elements of
$\{\North, \East, \South, \West\}$. In case $\ell_1$ and $\ell_2$ are
parallel, we assume that $\ell_1$ lies below $\ell_2$, and thus
$\W_B=\East$.

\subparagraph{Duality.}  We make frequent use of the
standard point-line duality~\cite{deberg08computational_geometry}, where we map objects
in \emph{primal} space to objects in a \emph{dual} space. In
particular, a primal point $p=(a,b)$ is mapped to the dual line
$p^*: y = ax - b$ and a primal line $\ell: y = ax + b$ is mapped to the
dual point $\ell^* = (a,-b)$. If in the primal a point
$p$ lies above a line $\ell$, then in the dual the line $p^*$ lies
below the point $\ell^*$.

For a set of points $P$ with duals $P^* = \{p^* \mid p \in P\}$, we are
often interested in the \emph{arrangement} $\A(P^*)$, i.e., the
vertices, edges, and faces formed by the lines in $P^*$. Two unbounded
faces of $\A(P^*)$ are \emph{antipodal} if their unbounded edges have
the same two supporting lines. Since
every line contributes to four unbounded faces, there are $O(n)$ pairs of antipodal faces. We denote the upper envelope of $P^*$, i.e., the polygonal chain following the highest line in $\A(P^*)$, by $\U(P^*)$, and the lower envelope by $\L(P*)$.

\section{Properties of an optimal separator.} 
\label{sec:propertiesOptimalSeparator}
Next, we prove some structural properties about the lines bounding the region $\W_B$ containing (most of the) the blue points in $B$. First for strips:

\begin{lemma}
  \label{lem:opt_strip}
    For the strip classification problem there exists an optimum where one line goes through two points and the other through at least one point.
\end{lemma}

\begin{proof}
    Clearly, we can shrink an optimal strip $\W_B(\ell_1,\ell_2)$ until both $\ell_1$ and $\ell_2$ contain a (blue) point, say $b_1$ and $b_2$, respectively. Now rotate $\ell_1$ around $b_1$ and $\ell_2$ around $b_2$ in counter-clockwise direction until either $\ell_1$ or $\ell_2$ contains a second point.
\end{proof}

Something similar holds for wedges:

\begin{lemma}
\label{lem:redBlueOptimum}
For any wedge classification problem there exists an optimum where both lines go through a blue and a red point.
\end{lemma}
\begin{proof}
We first show that any single wedge can be adjusted such that both its lines go through a blue and a red point, without misclassifying any more points. We then show the same for any double wedge. Since this also holds for any optimal wedge, we obtain the lemma.

\begin{claim}
\label{lem:lines_on_points_single}
Let
  $\W_B(\ell_1,\ell_2)$ be a single wedge so that there is at least
  one correctly classified point of each color. There exists a single wedge $\W_B(\ell'_1,\ell'_2)$
  such that: (1) both $\ell_1'$ and $\ell_2'$ go through a red point and a blue point,
  and (2) $\E(\ell_1',\ell_2') \subseteq \E(\ell_1,\ell_2)$.
\end{claim}
\begin{claimproof}
We show how to find $\ell_1'$ with a fixed $\ell_2$. Line $\ell_2'$ can be found in the same way afterwards while fixing $\ell_1'$. 

W.l.o.g. assume $\W_B$ is the $\West$ wedge. Let $B' \subseteq B$ be the correctly classified blue points in that wedge. Start with $\ell_1' = \ell_1$ and shift it downwards it until we hit the convex hull $\CH(B')$. Note that this does not violate (2): no extra red points are misclassified since we only make the $\West$ wedge smaller, and no extra blue points are misclassified because we stop at the first correctly classified one we hit. Rotate $\ell_1'$ clockwise around $\CH(B')$ until we hit a red point, at which point we satisfy (1). If $\ell_1'$ becomes vertical, the naming of the wedges shifts clockwise (e.g. the $\West$ wedge becomes the $\North$ wedge), so we must change $\W_B(\ell_1',\ell_2')$ and $\W_R(\ell_1',\ell_2')$ appropriately. If $\ell_1'$ becomes parallel to $\ell_2$, the $\East$ wedge (temporarily) becomes a strip and the $\West$ wedge disappears. Immediately afterwards the strip becomes the $\West$ wedge, and a new empty $\East$ wedge appears. If $B$ is assigned $\West$ or $\East$ at this time we must change $\W_B(\ell_1',\ell_2')$ and $\W_R(\ell_1',\ell_2')$ appropriately.

This procedure does not violate (2) because all of $B'$ lies on the same side of $\ell_1'$ at all times, and we never cross red points. It terminates, i.e., we hit a red point before having rotated around the entire convex hull, because we assumed there to be at least one correctly classified red point which must lie outside the $\West$ wedge and therefor outside of $\CH(B')$.
\end{claimproof}

Now we show the same for double wedges:

\begin{claim}
\label{lem:lines_on_points_double}
Let
  $\W_B(\ell_1,\ell_2)$ be a double wedge so that there is at least
  one correctly classified point of each color. There exists a double wedge $\W_B(\ell'_1,\ell'_2)$
  such that: (1) both $\ell_1'$ and $\ell_2'$ go through a red point and a blue point,
  and (2) $\E(\ell_1',\ell_2') \subseteq \E(\ell_1,\ell_2)$.
\end{claim}
\begin{claimproof}
We show how to find $\ell_1'$ with a fixed $\ell_2$. Line $\ell_2'$ can be found in the same way afterwards while fixing $\ell_1'$. 

W.l.o.g. assume $\W_B$ is the bowtie consisting of the $\West$ and $\East$ wedges. Consider the dual arrangement $\A(B^*,R^*)$, where we want segment $\ol{{\ell_1^*} \ell_2^*}$ to intersect blue lines but not red lines. Let $m$ be the supporting line of segment $\ol{{\ell_1^*} \ell_2^*}$. Start with ${\ell_1^*}' = \ell_1^*$. Note that if ${\ell_1^*}'$ ever lies in a face with red and blue segments on its boundary we are done: set ${\ell_1^*}'$ as one of the red-blue intersections, which satisfies (1) and does not change (2). Otherwise we distinguish two cases:
\begin{itemize}
    \item ${\ell_1^*}'$ lies in an all red face. We can shrink $\ol{{\ell_1^*}' \ell_2^*}$ by moving ${\ell_1^*}'$ towards $\ell_2^*$ along $m$ until we enter a bicolored face, in which case we are done. This will not violate (2), since shrinking the segment can only cause it to intersect fewer red lines. We must enter a bicolored face before the segment collapses since we assumed there to be at least one correctly classified blue point. See Figure \ref{fig:redBlueLine}a.
    
    \item ${\ell_1^*}'$ lies in an all blue face. We extend $\ol{{\ell_1^*}' \ell_2^*}$ by moving ${\ell_1^*}'$ along $m$ until either (i) ${\ell_1^*}'$ enters a bicolored face, in which case we are done, or until (ii) ${\ell_1^*}'$ ends up in an outer face which is unbounded in the direction of $m$. This does not violate (2), since extending the segment will only make it intersect more blue lines. See Figure \ref{fig:redBlueLine}b.
    
    In case (ii), let $F$ be the outer face ${\ell_1^*}'$ ends up in, and let $p^*$ be some point on $m$ in $F$. Let $F'$ be the antipodal face of $F$, and let $q^*$ be some point on $m$ in $F'$. See Figure \ref{fig:redBlueLine}c. Observe that segment $\ol{q^* \ell_2^*}$ intersects exactly those lines that segment $\ol{p^* \ell_2^*}$ does not intersect, and the other way around. In the primal, points in the hourglass wedge of $(p,\ell_2)$ are in the bowtie wedge of $(q,\ell_2)$. Therefore segment $\ol{p^* \ell_2^*}$ yields the exact same classification $\ol{q^* \ell_2^*}$, after assigning $B$ to the hourglass wedge instead of the bowtie wedge.
    
    This means we can set ${\ell_1^*}' = q$, change the color assignment appropriately, and shrink segment $\ol{{\ell_1^*}' \ell_2^*}$ until ${\ell_1^*}'$ lies in a bicolored face. This does not violate (2), since shrinking the segment only makes it intersect fewer blue lines. We must enter a bicolored face before the segment collapses since we assumed there to be at least one corrrectly classified red point.
    \claimqedhere
\end{itemize}
\end{claimproof}

The above two claims tell us that, given some optimal (double) wedge for a wedge separation problem, we can adjust the wedge until both lines go through a red and a blue point, proving the lemma.
\end{proof}

\begin{figure}[tb]
    \centering
    \includegraphics{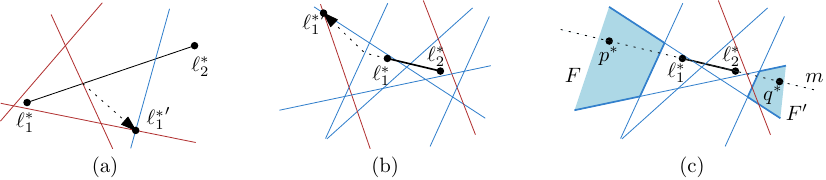}
    \caption{(a)/(b): shrinking/extending segment $\ol{{\ell_1^*} \ell_2^*}$ until it reaches a bicolored face. (c): points $p^*$ and $q^*$ in antipodal outer faces $F$ and $F'$. Segment $\ol{p^* \ell_2^*}$ intersects exactly those lines that $\ol{q^* \ell_2^*}$ does not intersect. }
    \label{fig:redBlueLine}
\end{figure}

\subsection{Simple general algorithm}
\label{sub:brute_force}
Lemma~\ref{lem:redBlueOptimum} tells us we only have to consider lines through red and blue points. Hence, there is a simple brute-force algorithm that considers all pairs of such lines, which works for both wedges and double wedges and any type of outliers.

\begin{theorem}
    Given two sets of $n$ points $B, R \subset \R^2$, we can construct a (double) wedge $\W_B$ minimizing either $k_r$, $k_b$, or $k$ in $O(n^4)$ time.
\end{theorem}
\begin{proof}
By Lemma~\ref{lem:redBlueOptimum} we only have to consider lines through blue and red points. There are $O(n^2)$ such lines, so $O(n^4)$ pairs of such lines. We could trivially calculate the number of misclassifications for two fixed lines in $O(n)$ time by iterating through all points, which would result in $O(n^5)$ total time, but we can improve on this.

Construct the dual arrangement $\A(B^* \cup R^*)$ of $B \cup R$ in $O(n^2)$ time. A red-blue intersection in the dual $\A(B^* \cup R^*)$ corresponds to a candidate line through a red and a blue point in the primal. Choose two arbitrary red-blue vertices as $\ell_1^*$ and $\ell_2^*$, and calculate the number of red and blue points in each of the four wedge regions in $O(n)$ time. Move $\ell_1^*$ through the arrangement in a depth-first search order, updating the number of points in each wedge at each step. There is only a single point that lies on the other side of $\ell_1$ after this movement, so this update can be done in constant time. After every step of $\ell_1^*$, also move $\ell_2^*$ through the whole arrangement in depth-first search order, updating the number of points in each wedge, again in constant time. Finally, output the pair of lines that misclassify the fewest points. There are $O(n^2)$ choices for $\ell_1^*$, and for each of those there are $O(n^2)$ choices for $\ell_2^*$. Since every update takes constant time, this takes $O(n^4)$ time in total.
\end{proof}

\section{Separation with a line: one-sided outliers}
\label{sec:Single_line_separation_one-sided}
In this section, we consider the case when $\W_B$ is an upper halfplane $\ell^+$ bounded by line $\ell$, and we minimize the number of red outliers $k_R$. Minimizing $k_B$ can be done symmetrically.

We solve this problem in the dual, where our goal is to find a
point $\ell^*$ that lies above all blue lines and above as few red lines as possible. See
Figure~\ref{fig:single_line_red_outliers}. Since $\ell$ lies above all blue lines, it lies above $\U(B^*$). In fact, we can assume $\ell$ lies on $\U(B^*)$, since we can always shift it downwards until it
lies on $\U(B^*)$ without increasing the number of red lines below
$\ell^*$.

\begin{figure}[tb]
  \centering
  \includegraphics{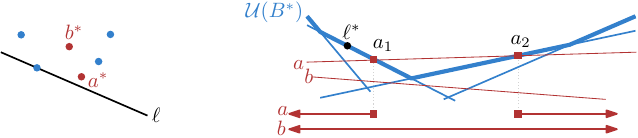}
  \caption{(left) Line $\ell$ misclassifies red points $a$ and $b$. (right) In the dual
    space that means $\ell$ lies in the forbidden regions of both $a$ and $b$.}
  \label{fig:single_line_red_outliers}
\end{figure}

As $\U(B^*)$ is $x$-monotone, there is only one degree of freedom for choosing point $\ell^*$: its $x$-coordinate. We parameterize $\U(B^*)$ over $\R$, our \emph{parameter space}, such that each value $p \in \R$ corresponds to the point $\ell^*$ on $\U(B^*)$ at $x = p$. We wish to find a value in this parameter space whose corresponding point minimizes the number of red misclassifications, i.e., the number of red lines below $\ell^*$. Let the \emph{forbidden regions} of a red line $r \in R^*$ be those intervals in the parameter space in which corresponding points misclassify $r$.

Since $\U(B^*)$ is a convex polygonal chain, line $r$ intersects $\U(B^*)$ at most two times. Therefore, we distinguish between two types of red lines:

\begin{itemize}
\item Line $a$ intersects $\U(B^*)$ in $a_1$ and $a_2$, with $a_1 \leq a_2$. Points $\ell^*$ left of $a_1$ or right of $a_2$ misclassify $a$, so $a$ produces two forbidden intervals: $(-\infty,a_1)$ and $(a_2,\infty)$.
\item Line $b$ does not intersect $\U(B^*)$. Then all points $\ell^*$ misclassify $a$, so $a$ produces one trivial forbidden interval: $(-\infty,\infty)$.
\end{itemize}

We can compute these $O(n)$ forbidden intervals in $O(n \log n)$ time (using that $\U(B^*)$ is convex). We can then compute a point $\ell^* \in \R$ with minimum ply in these forbidden regions by sorting and scanning the intervals, which also takes $O(n \log n)$ time. The primal line $\ell$ (corresponding to dual point $\ell^*$) is then a line above all blue lines and above as few red lines as possible. We conclude:

\begin{theorem}
  \label{thm:single_outliers_line}
    Given two sets of $n$ points $B, R \subset \R^2$, we can find a line with all points of $B$ on one side and as many points of $R$ as possible on the other side in $O(n \log n)$ time.
\end{theorem}

\section{Separation with a strip}
\label{sec:Separation_with_a_strip}

In this section we consider the case where lines $\ell_1$ and $\ell_2$
are parallel, with $\ell_2$ above $\ell_1$, and thus $\W_B(\ell_1,\ell_2)$ forms a strip. We want $B$ to be inside the strip, and $R$ outside. We work in the dual, where we want to find two points $\ell_1^*$
and $\ell_2^*$ with the same $x$-coordinate such that vertical segment
$\ol{\ell_1^* \ell_2^*}$ intersects the lines in $B^*$ but not the
lines in $R^*$.

\subsection{Strip separation with red outliers}
\label{sub:Strip_Red_outliers}

We first consider the case where all blue points must be correctly
classified, and we minimize the number of red outliers $k_R$. We present an $O(n\log n)$ time algorithm to this end. Note that this runtime matches the existing algorithm from
Seara~\cite{seara2002geometric}.
We wish to find a segment $\ol{\ell_1^* \ell_2^*}$ that intersects all lines in $B^*$,
so $\ell_1^*$ must be above the upper envelope $\U(B^*)$ and $\ell_2^*$ must be below the lower envelope
$\L(B^*)$. Again we can assume $\ell_1^*$ to lie on $\U(B^*)$ and
$\ell_2^*$ on $\L(B^*)$, since shortening $\ol{\ell_1^* \ell_2^*}$ can only
decrease the number of red lines intersected. 

As before there is only one degree of
freedom for choosing our segment: its $x$-coordinate.
We parameterize $\U(B^*)$ and $\L(B^*)$ over our parameter space $\R$, such that each point $p \in \R$ corresponds to the vertical segment $\ol{\ell_1^* \ell_2^*}$ on the line $x = p$.
We wish to find a point in this parameter space whose corresponding segment minimizes the number of red misclassifications, i.e., the number of red intersections. 
We distinguish between four types of red lines, as in Figure \ref{fig:stripLineTypes}:

\begin{itemize}
    \item Line $a$ intersects $\U(B^*)$ in points $a_1$ and $a_2$, with $a_1 \leq a_2$. Segments with $\ell^*_1$ left of $a_1$ or right of $a_2$ misclassify $a$, so $a$ produces two forbidden intervals: $(-\infty,a_1)$ and $(a_2,\infty)$. 
    
    \item Line $b$ intersects $\L(B^*)$ in points $b_1$ and $b_2$, with $b_1 \leq b_2$. Similar to line $a$ this produces forbidden intervals $(-\infty,b_1)$ and $(b_2,\infty)$.

    \item Line $c$ intersects $\L(B^*)$ in $c_1$ and $\U(B^*)$ in $c_2$. Only segments between $c_1$ and $c_2$ misclassify $c$. This gives one forbidden interval: $(\min\{c_1,c2\},\max\{c_1,c_2\})$. 

    \item Line $d$ intersects neither $\U(B^*)$ nor $\L(B^*)$. All segments misclassify $d$. This gives one trivial forbidden region, namely the entire space $\R$.
\end{itemize}

The above list is exhaustive. Clearly a line can not intersect $\U(B^*)$ or $\L(B^*)$ more than twice. Let $b_1, b_2$ be the two blue lines supporting the unbounded edges of $\U(B^*)$, and note that $b_1$ and $b_2$ are also the supporting lines of the unbounded edges of $\L(B^*)$. Therefore, if a line intersects $\U(B^*)$ twice it can not intersect $\L(B^*)$ and vice versa. Additionally, any line intersecting $\U(B^*)$ once must have a slope between those of $b_1$ and $b_2$, hence it must also intersect $\L(B^*)$ once, and vice versa.

\begin{figure}[tb]
    \centering
    \includegraphics{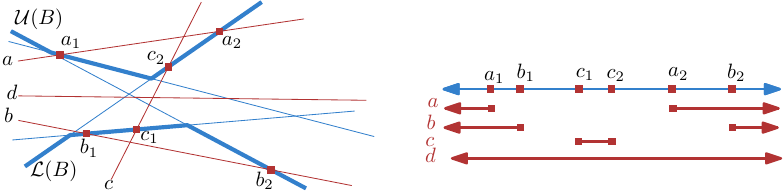}
    \caption{Four types of red lines for strip separation, with restrictions on their parameter space.}
    \label{fig:stripLineTypes}
\end{figure}

Computing $\U(B^*)$ and $\L(B^*)$ takes $O(n \log n)$ time. Given a red line $r \in R^*$ we can compute its intersection points with $\U(B^*)$ and $\L(B^*)$ in $O(\log n)$. Computing the forbidden regions thus takes $O(n\log n)$ time in total. Recall that our goal is to find a point with minimum ply in these forbidden regions. We can compute such a point in $O(n \log n)$ time by sorting and scanning the intervals. We conclude:

\begin{theorem}
    Given two sets of $n$ points $B, R \subset \R^2$, we can construct a strip $\W_B$ minimizing the number of red outliers $k_R$ in $O(n \log n)$ time.
\end{theorem}

\subsection{Strip separation with blue outliers}
\label{sub:Strip_Blue_outliers}
We now consider the case when all red points must be correctly classified, and we minimize the number of blue outliers $k_B$. Seara~\cite{seara2002geometric} uses a very similar algorithm to find the minimum number of strips needed to perfectly separate $B$ and $R$.

We are looking for a strip of two lines $\ell_1$ and $\ell_2$ containing no red points and as many blue points as possible. In the dual this is a vertical segment $\ol{\ell_1^*\ell_2^*}$ intersecting no red lines and as many blue lines as possible. That means the segment must lie in a face of $\A(R^*)$; similar to before we can always extend a segment until its endpoints lie on red lines.

Say we wish to find the best segment at a fixed $x$-coordinate, so on a vertical line $z$. Line $z$ is divided into $m+1$ intervals by the $m$ red lines, where each interval is a possible segment. This segment intersects exactly those blue lines that intersect $z$ in the same interval, so we are looking for the red interval in which the most blue intersections lie.

\subparagraph{Algorithm.}
Calculate all $O(n^2)$ intersections between lines in $B^* \cup R^*$, and sort them. We sweep through them from left to right with a vertical line $z$. At any time, there are $m+1$ red intervals on the sweepline. Number the intervals $0$ to $m$ from bottom to top. We maintain a list $S$ of size $m+1$, such that $S[i]$ contains the number of blue lines intersecting $z$ in interval $i$. Additionally, for every red line $r_j$ we maintain the (index of the) interval $a_j$ above it. There are 3 types of events:

\begin{itemize}
    \item Red-red intersection between lines $r_{j}$ and $r_{k}$, with the slope of $r_{j}$ larger than that of $r_{k}$. This means red interval $a_{j}$ collapses and opens again. We adjust the adjacent intervals of both lines accordingly, by incrementing $a_{j}$ and decrementing $a_{k}$.
    \item Blue-blue intersection: two blue lines change places in an interval, but the number of blue lines in the interval stay the same, so we do nothing.
    \item Red-blue intersection between red line $r_j$ and blue line $b$. Line $b$ moves from one interval to an adjacent one. Specifically, if the slope of $r_j$ is larger than that of $b$ we decrement $S[a_j]$ and increment $S[a_j - 1]$, and otherwise we increment $S[a_j]$ and decrement $S[b_j - 1]$. 
\end{itemize}

Each event is handled in constant time. Sorting the events takes $O(n^2 \log n)$ time.

\begin{theorem}
    Given two sets of $n$ points $B, R \subset \R^2$, we can construct a strip containing the most points of $B$ and no points of $R$ in $O(n^2 \log n)$ time.
\end{theorem}

\subsection{Strip separation with both outliers}
\label{sub:strip_Both_outliers}
Finally we consider the case where we allow both red and blue outliers, and we minimize the total number of outliers $k$. We again consider the dual in which $\W_B(\ell_1,\ell_2)$ corresponds to
a vertical segment $\overline{\ell_1^*\ell_2^*}$. By
Observation~\ref{lem:opt_strip} there is an optimal solution where: (i)
$\ell_2^*$ is a vertex of $\A(B^* \cup R^*)$ and $\ell_1^*$ lies on a
line from $B^* \cup R^*$ above $\ell^*_2$, or (ii) vice versa. We
present an $O(n^2\log n)$ time algorithm to find the best solution of
type (i). Computing the best solution of type (ii) is analogous.

We again sweep the arrangement $\A(B^* \cup R^*)$ with
a vertical line. During the sweep we maintain a data structure storing
the lines intersected by the sweep line in bottom-to-top order, so
that given a vertex $\ell_2^*$ on the sweepline we can efficiently
find a corresponding point $\ell_1^*$ above $\ell_2^*$ for which
$|\E(\ell_1,\ell_2)|$ is minimized. In particular, we argue that
we can answer such queries in $O(\log n)$ time, and support updates
(insertions and deletions of lines) in $O(\log n)$ time. It then
follows that we obtain an $O(n^2\log n)$ time algorithm by performing
$O(1)$ updates and one query at every vertex of $\A(B^*\cup R^*)$.

\subparagraph{Finding an optimal line $\ell_1$.} Fix a point
$\ell_2^*=(x,y_2)$, and consider the number of blue outliers
$k_B(y_1) = |\E_B(\ell_1,\ell_2)|$ in a strip with
$\ell_1^*=(x,y_1)$. Observe that $k_B(y_1)$ is the number of blue
lines passing below $\ell_2^*$ plus the number of blue lines passing
above $\ell^*_1$. Hence $k_B(y_1)$ is a non-increasing piecewise
constant function of $y_1$. Analogously, the number of red outliers
$k_R(y_1)$ is the number of red lines passing in between $\ell_1$ and
$\ell_2$. This function is non-decreasing piecewise constant function
of $y_1$. See Figure~\ref{fig:strip_both_outliers}. We have
$k(y_1)=k_R(y_1)+k_B(y_1)$, and we are interested in the value
$\hat{y} = \arg\min_{y_1} k(y_1)$ where $k$ attains its minimum.

\begin{figure}[tb]
  \centering
  \includegraphics{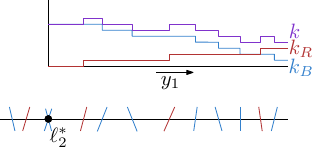}
  \caption{A snapshot of the sweepline (rotated to be horizontal), and
    the functions $k_B$, $k_R$, and $k$ expressing the number of
    outliers as a function of $y_1$.
  }
  \label{fig:strip_both_outliers}
\end{figure}

\subparagraph{The data structure.} We now argue we can maintain an
efficient representation of the function $k$ in case
$\ell_{2,y}^* = -\infty$. We then argue that we can also use the
structure to query with other values of $\ell_2^*$. Our data structure
is a fully persistent red-black tree~\cite{driscoll89makin_data_struc_persis} that stores the lines of
$B^* \cup R^*$ in the order in which they intersect the vertical
(sweep)line at $x$. We annotate each node $\nu$ with: (i) the number
$k_R^\nu$ of red lines in its subtree, (ii) the number $k_B^\nu$ of
blue lines in its subtree, (iii) the minimum value $\min^\nu$ of
$k(y_1)$ when restricted to all lines in the subtree rooted at $\nu$,
and (iv) the value $\hat{y}^\nu$ achieving that minimum. Let $\ell$
and $r$ be the children of $\nu$, and observe that
$\min^\nu = \min \{ \min^\ell + k_B^r , \min^r + k_R^\ell \}$. Hence,
$\min^\nu$ can be (re)computed from the values of its children. The
same applies for $k_R^\nu$ and $k_B^\nu$. Therefore, we can easily
support inserting or deleting a line in $O(\log n)$ time. Indeed,
inserting a red line that intersects the vertical line at $x$ in $y$,
increases the error either for all values $y' > y$ or for all value
$y' < y$ by exactly one, hence this affects only $O(\log n)$ nodes in
the tree.

Observe that for $\ell^*_{2,y}= -\infty$ the root $\nu$ of the tree
stores the value $\min^\nu = \min_y k(y)$, and the value $\hat{y}^\nu$
attaining this minimum. Hence, for such queries we can report the
answer in constant time. To support querying with a different value of
$\ell^*_{2,y}$, we simply split the tree at $\ell^*_{2,y}$, and use
the subtree storing the lines above $\ell^*_2$ to answer the
query. Observe that the number of blue lines below $\ell_2^*=(x,y_2)$
is a constant with respect to $y_1 \geq y_2$. Hence, it does not
affect the position at which $\min_{y > y_2} k(y)$ attains its
minimum. Splitting the tree and then answering the query takes
$O(\log n)$ time. After the query we discard the two subtrees and
resume using the original one, which we still have access to as the
tree is fully persistent.
We thus obtain the following result:

\begin{theorem}
  \label{thm:opt_strip}
    Given two sets of $n$ points $B, R \subset \R^2$, we can compute a strip $\W_B$ minimizing the total number of outliers $k$ in $O(n^2 \log n)$ time.
\end{theorem}

\section{Separation with a wedge}
\label{sec:single_wedge}
We consider the case where the region $\W_B$ is a single wedge and $\W_R$ is the other three
wedges. In Sections \ref{sub:Single_wedge_with_red_outliers}, \ref{sub:Single_wedge_with_blue_outliers}, and \ref{sub:single_wedge_both_outliers} we show how to minimize $k_R$, $k_B$, and $k$, respectively.

\subsection{Wedge separation with red outliers}
\label{sub:Single_wedge_with_red_outliers}

We distinguish between $\W_B$ being an $\East$ or $\West$ wedge, and a $\North$ or $\South$ wedge. In either case we can compute optimal lines $\ell_1$ and $\ell_2$ defining $\W_B$ in $O(n\log n)$ time.

\subparagraph{Finding an East or West wedge.}
We wish to find two lines $\ell_1$ and $\ell_2$ such that every blue point and as few red points as possible lie above $\ell_1$ and below $\ell_2$.
In the dual this corresponds to two points $\ell_1^*$ and $\ell_2^*$ such that all blue lines and as few red lines as possible lie below $\ell_1^*$ and above $\ell_2^*$, as in Figure~\ref{fig:arrangement}.

\begin{figure}[tb]
    \centering
    \includegraphics{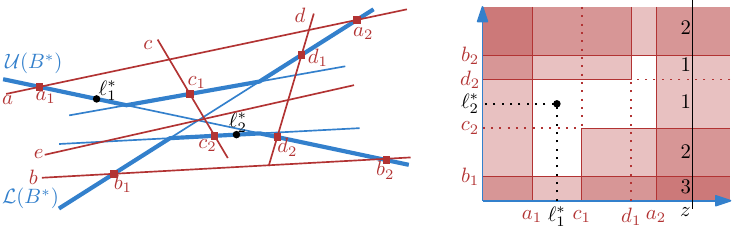}
    \caption{The arrangement of $B^* \cup R^*$ with its parameter space and forbidden regions.}
    \label{fig:arrangement}
\end{figure}

Clearly $\ell_1^*$ must lie above $\U(B^*)$, and $\ell_2^*$ below $\L(B^*)$, and again we can assume they lie on $\U(B^*)$ and $\L(B^*)$, respectively. We now have two degrees of freedom, one for choosing $\ell_1^*$ and one for choosing $\ell_2^*$.
Again we parameterize $\U(B^*)$ and $\L(B^*)$, but this time over $\R^2$, such that a point $(p,q)$ in this parameter space corresponds to two dual points $\ell_1^*$ and $\ell_2^*$, with $\ell_1^*$ on $\U(B^*)$ at $x = p$ and $\ell_2^*(y)$ on $\L(B^*)$ at $x = q$, as illustrated in Figure \ref{fig:arrangement}.
We wish to find a value in our parameter space whose corresponding segment minimizes the number of red misclassifications. 
Recall the forbidden regions of a red line $r$ are those regions in the parameter space in which corresponding segments misclassify $r$. We distinguish between five types of red lines, as in Figure \ref{fig:arrangement}:

\begin{itemize}
    \item Line $a$ intersects $\U(B^*)$ in points $a_1$ and $a_2$, with $a_1$ left of $a_2$.
    Only segments with $\ell_1^*$ left of $a_1$ or right of $a_2$ misclassify $a$.
    This produces two forbidden regions: $(-\infty, a_1) \times (-\infty, \infty)$ and $(a_2, \infty) \times (-\infty, \infty)$.

    \item Line $b$ intersects $\L(B^*)$ in points $b_1$ and $b_2$, with $b_1$ left of $b_2$. Symmetric to line $a$ this produces forbidden regions $(-\infty, \infty) \times (-\infty, b_1)$ and $(-\infty, \infty) \times (b_2, \infty)$.

    \item Line $c$ intersects $\U(B^*)$ in $c_1$ and $\L(B^*)$ in $c_2$, with $c_1$ left of $c_2$. Only segments with endpoints after $c_1$ and before $c_2$ misclassify $c$. This produces the region $(c_1, \infty) \times (-\infty, c_2)$. (Segments with endpoints before $c_1$ and after $c_2$ do intersect $c$, but do not misclassify it.)
    
    \item Line $d$ intersects $\U(B^*)$ in $d_1$ and $\L(B^*)$ in $d_2$, with $d_1$ right of $d_2$. Symmetric to line $c$ it produces the forbidden region $(-\infty, d_1) \times (d_2, \infty)$.

    \item Line $e$ intersects neither $\U(B^*)$ nor $\L(B^*)$. All segments misclassify $e$. In the primal this corresponds to red points inside the blue convex hull. This produces one forbidden region; the entire plane $\R^2$.
\end{itemize}

As in Section \ref{sub:Strip_Red_outliers}, the above list is exhaustive.

The forbidden regions generated by the red lines $r^* \in R^*$ divide the parameter space in axis-aligned orthogonal regions. 
Our goal is again to find a point with minimum ply in these forbidden regions. For this we prove the following lemma:

\begin{lemma}
\label{lem:minimum_ply}
    Given a set $\mathcal{R}$ of $n$ constant complexity, axis-aligned, orthogonal regions, we can compute the point with minimum ply in $O(n \log n)$ time.
\end{lemma}
\begin{proof}
    We sweep through the plane with a vertical line $z$ while maintaining a minimum ply point on $z$. See Figure \ref{fig:arrangement} (right) for an illustration.
    As a preprocessing step, we cut each region into a constant number of axis-aligned (possibly unbounded) rectangles, and build the skeleton of a segment tree on the $y$-coordinates of the vertical sides of these rectangles in $O(n \log n)$ time~\cite{deberg08computational_geometry}.
    This results in a binary tree with a leaf for each elementary $x$-interval induced by the segments. A node $v$ corresponds to the union of the intervals of its children. The canonical subset of $v$ is the set of intervals containing $v$ but not the parent of $v$. For a node $v$ we store the size $s(v)$ of its canonical subset, the minimum ply $\ply(v)$ within the subtree of $v$, and a point attaining this minimum ply.
    
    We start with $z$ at $-\infty$ and sweep to the right. When we encounter the left (respectively right) side of a rectangle with vertical segment $I = (y_1,y_2)$, insert (respectively delete) $I$ in the segment tree. Since we already constructed the skeleton of this tree, the endpoints of $I$ are already present, so the shape of the tree does not change. Updating $s(v)$ takes $O(\log n)$ time, since $I$ is in the canonical subset of only $O(\log n)$ nodes. The minimum ply in a node $v$ with children $c_1,c_2$, and a point attaining this minimum, can be updated simultaneously, in a bottom-up order: $\ply(v) = s(v) + \min(\ply(c_1), \ply(c_2))$. After every update, the root node stores the current minimum ply. We maintain and return the overall minimum ply over all positions of the sweepline.
    
    Since there are $O(n)$ rectangles, each of which is added and removed once in $O(\log n)$ time, this leads to a running time of~$O(n \log n)$.
\end{proof}

We construct $\U(B^*)$ and $\L(B^*)$ in $O(n \log n)$ time. For every red line $r$, we calculate its intersections with $\U(B^*)$ and $\L(B^*)$ in $O(\log n)$ time, determine its type ($a-e$), and construct its forbidden regions. By Lemma~\ref{lem:minimum_ply} we can find a point with minimum ply in these forbidden regions in $O(n \log n)$ time. 

\begin{lemma}
    Given two sets of $n$ points $B, R \subset \R^2$, we can construct an \East or \West wedge containing all points of $B$ and the fewest points of $R$ in $O(n \log n)$ time.
\end{lemma}

\subparagraph{Finding a North or South wedge}
We wish to find two lines $\ell_1$ and $\ell_2$ such that $\W_B(\ell_1,\ell_2) = \South(\ell_1,\ell_2)$, i.e., such that every blue point and as few red points as possible lie below both $\ell_1$ and $\ell_2$. The case where $\W_B(\ell_1,\ell_2) = \North(\ell_1,\ell_2)$ is symmetric. In the dual this corresponds to two points $\ell_1^*$ and $\ell_2^*$ such that all blue lines and as few red lines as possible lie above both $\ell_1^*$ and $\ell_2^*$.

Clearly both $\ell_1^*$ and $\ell_2^*$ must lie below $\L(B^*)$, and again we can assume they lie on $\L(B^*)$. 
Similar to before we parameterize the $x$-coordinate of both points over $\R^2$, such that a point $(p,q)$ in the parameter space corresponds to two dual points $\ell_1^*$ and $\ell_2^*$, with $\ell_1^*$ on $\L(B^*)$ at $x = p$ and $\ell_2^*$ on $\L(B^*)$ at $x = q$.
Note that the resulting parameter space~$P$ is symmetric over $y = x$, since $\ell_1^*$ and $\ell_2^*$ are interchangeable. 

Let $s_{\min}$ ($s_{\max}$) be the minimum (maximum) slope of all blue lines.
There are now four types of red lines, as illustrated in Figure~\ref{fig:arrangementUpwards}:

\begin{itemize}
    \item line $a$: intersects $\L(B^*)$ twice in $a_1$ and $a_2$.
    Line $a$ is misclassified if both $\ell_1^*$ and $\ell_2^*$ lie below $a$.
    In the parameter space, this corresponds to four forbidden corners of the parameter space, as in Figure~\ref{fig:arrangementUpwards}.
    \item line $b$: intersects $\L(B^*)$ once in $b_1$ and has a slope~$s\in(-\infty,s_{\min})$.
    Only segments with both endpoints left of $b_1$ misclassify $b$, producing a forbidden bottomleft quadrant in the parameter space~$P$.
    \item line $c$: intersects $\L(B^*)$ once in $c_1$ and has a slope~$s\in(s_{\max},\infty)$.
    Similar to $b$, this produces a forbidden topright quadrant.
    \item line $d$: does not intersect $\L(B^*)$.
    This point will always be misclassified by a \North wedge.
    In the primal, this corresponds to red points lying in or directly below the convex hull of the blue points.
\end{itemize}

\begin{figure}[tb]
    \centering
    \includegraphics{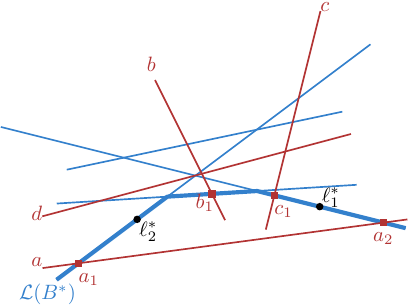}\hfil
    \includegraphics{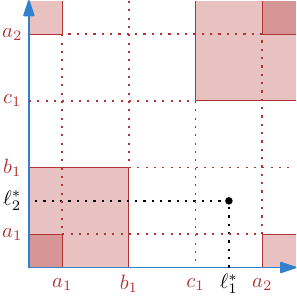}
    \caption{Left: The arrangement of lines $B^*$ and $R^*$. Right: The corresponding parameter space~$P$ and forbidden regions.}
    \label{fig:arrangementUpwards}
\end{figure}

As before we construct all forbidden regions, and apply the algorithm of Lemma~\ref{lem:minimum_ply} to obtain a point in the parameter space with minimum ply in $O(n \log n)$ time.

\begin{lemma}
    Given two sets of $n$ points $B, R \subset \R^2$, we can construct a $\North$ or $\South$ wedge containing all points of $B$ and the fewest points of $R$ in $O(n \log n)$ time.
\end{lemma}

\begin{theorem}
    Given two sets of $n$ points $B, R \subset \R^2$, we can construct a wedge containing all points of $B$ and the fewest points of $R$ in $O(n \log n)$ time.
\end{theorem}

\subsection{Wedge separation with blue outliers}
\label{sub:Single_wedge_with_blue_outliers}
We now consider the case where all red points must be classified
correctly, and we minimize the number of blue outliers $k_B$. We show
how to find an optimal \North wedge; finding optimal \South, \East, or
\West wedges can be done analogously.

\begin{figure}[tb]
    \centering
    \includegraphics{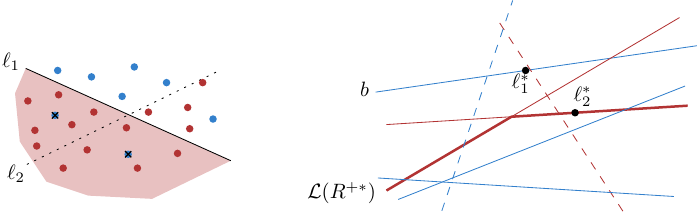}
    \caption{Left: in the primal we need to consider only the points above $\ell_1$. Right: in the dual we need to consider only the lines below $\ell_1^*$. In particular, $\ell_1^*$ should lie below (on) $\L(R^{+*})$.}
    \label{fig:wedgeBlueOutliersDual}
\end{figure}

Fix line $\ell_1$ and consider the problem of finding an optimal
corresponding line $\ell_2$. All points below $\ell_1$ lie outside the
\North wedge, regardless of our choice of $\ell_2$, and thus we have
to consider only the points $B^+ \subseteq B$ and $R^+ \subseteq R$
above $\ell_1$. See the left of Figure \ref{fig:wedgeBlueOutliersDual}. We
do not allow red points in the \North wedge, so $\ell_2$ must lie
above all red points $R^+$ and below as many blue points $B^+$ as
possible. This is exactly the halfplane separation problem with blue
outliers we solved in Section
\ref{sec:Single_line_separation_one-sided} in $O(n \log n)$ time. We
can iterate through all $O(n^2)$ options for $\ell_1$ (by walking
through $\A(B^* \cup R^*)$) and compute the corresponding line
$\ell_2$ in $O(n \log n)$ time, which would lead to an $O(n^3 \log n)$
time algorithm. Below we describe an algorithm that avoids recomputing
$\ell_2$ from scratch every time, giving an $O(n^{5/2}\log n)$
time algorithm. We first give an overview.

Let $L$ be a set of lines, and let the \emph{level} $\level_L(p)$ of a
point $p$ (with respect to $L$) be the number of lines of $L$ that lie
below $p$. We define the level
$\level_L(s) = \max_{p \in s} \level_L(p)$ of a segment $s$ (with
respect to $L$) as the maximum level of any point $p$ on $s$.

Consider the dual, where we are looking for two points $\ell_1^*$ and
$\ell_2^*$ such that no red line and as many blue lines as possible
lie below both $\ell_1^*$ and $\ell_2^*$. By Lemma
\ref{lem:redBlueOptimum} we can assume both $\ell_1^*$ and $\ell_2^*$
lie on a red-blue intersection.  For a fixed point $\ell_1^*$ we are
interested only in the set of lines $B^{+*}$ and $R^{+*}$ below
$\ell_1^*$, and since $\ell_2^*$ must lie below all of $R^{+*}$
we can assume it lies on its lower envelope $\L(R^{+*})$. See the right of Figure
\ref{fig:wedgeBlueOutliersDual}. The wedge $\North(\ell_1, \ell_2)$
correctly classifies exactly $\level_{B^{+*}}(\ell_2^*)$ points. We
are thus looking for the pair of points $\ell_1^*$ and $\ell_2^*$ that
maximize the level $\level_{B^{+*}}(\ell_2^*)$.

We now show that we can compute $\ell_2^*$ efficiently for every candidate point
$\ell_1^*$, provided that there is an oracle that can answer (a batch
of) the following queries: given a point $\ell_1^*$ and a line segment $s$
lying on a red line, compute the level $\level_{B^{+*}}(s)$. We then
show that we can implement an oracle that answers all $O(n^2)$ queries
in $O(n^{5/2}\log n)$ time. This yields an
$O(n^{5/2}\log n)$ time algorithm to compute an optimal north wedge.

\subparagraph{Using an oracle to maintain $\ell_2^*$.}
Consider any blue line $b \in B$ and assume w.l.o.g. that it is horizontal. We will shift $\ell_1^*$ from left to right along $b$,
maintaining the set of red lines $R^{+*}$ below $\ell_1^*$. During
this shift $\ell_1^*$ crosses each of the other $n-1$ lines at most
once. We wish to maintain $\L(R^{+*})$ and a point with maximum level
w.r.t. $B^{+*}$ over all edges of $\L(R^{+*})$. Such a point corresponds to an optimal second point $\ell_2^*$
for the current point $\ell_1^*$. By repeating this shift for every
blue line $b \in B^*$ we consider all $O(n^2)$ candidate points
for $\ell_1^*$ and their corresponding optimal point $\ell_2^*$. This
thus allows us to report an optimal solution.

We first show that, while shifting $\ell_1^*$ along $b$, the number of
changes to $\L(R^{+*})$ is small. This allows us to explicitly
maintain the edges of $\L(R^{+*})$ in the leaves of a binary tree
(ordered on increasing $x$-coordinate). We refer to this tree as
the \emph{explicit tree} of $\L(R^{+*})$. We then augment this
explicit tree to additionally maintain the maximum level over all its edges.

\begin{lemma}
  \label{lem:setESizeN}
  Let $E$ be the union of all red edges that ever appear on
  $\L(R^{+*})$ while shifting $\ell_1^*$ along $b$. In other words,
  an edge $e$ is in $E$ iff there is a point $\ell_1^*$ on $b$ such
  that $e$ appears on $\L(R^{+*})$. This set $E$ has size $O(n)$.
\end{lemma}

\begin{proof}
  We partition the set of red lines $R^*$ into two sets of lines $P^*$
  and $G^*$, so that $P^*$ contains the lines with positive slope (we will refer to them as the Pink lines) and $G^*$ contains the lines
  with negative slope (the Green lines). See Figure
  \ref{fig:wedgeBlueGreenPink}. As before, we define $G^{+*}$ and
  $P^{+*}$ to be the green, respectively pink, lines below
  $\ell_1^*$. Let $E_G$ and $E_P$ be the set of edges that ever appear
  on $\L(G^{+*})$ and $\L(P^{+*})$, respectively.

\begin{figure}[tb]
    \centering
    \includegraphics{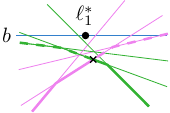}
    \caption{The set $R^*$ is partitioned into green (negative slope)
      and pink (positive slope) lines. Their lower envelopes intersect
      at most once.}
    \label{fig:wedgeBlueGreenPink}
\end{figure}

We first show that $E_G$ has size $O(n)$.
Due to the green lines all having
negative slope, we only ever insert lines into $G^{+*}$ when shifting
$\ell_1^*$ from left to right.
An insertion of a line can create
at most three new edges: one on the line itself, and at most two (now shortened) edges
on other intersected lines. Since there are $O(n)$ line insertions, that means
there are $O(n)$ edges in $E_G$ in total.
It follows from a symmetric argument that $E_P$ has size $O(n)$.

Consider some point $\ell_1^*$ on $b$. Note that $\L(R^{+*})$ is the
concatenation of a (possibly empty) part of $\L(P^{+*})$, followed by
a (possibly empty) part of $\L(G^{+*})$ (these parts are only empty
if $P^{+*}$ or $G^{+*}$ themselves are empty). This follows from the
fact that the edges on $\L(R^{+*})$ appear in order of decreasing
slope. Therefore, all edges of $\L(R^{+*})$ are also an edge in
$\L(P^{+*})$ or $\L(G^{+*})$, except the two edges where
$\L(P^{+*})$ and $\L(G^{+*})$ intersect. Since there are only two
such intersecting edges for each position of $\ell_1^*$, there are
$O(n)$ such edges in $E$ in total, and all other edges from $E$ are
either in $E_G$ or in $E_P$. It follows that $E$ has size $O(n)$.
\end{proof}

\begin{lemma}
\label{lem:redEnvelopeON}
While shifting $\ell_1^*$ along $b$, there are $O(n)$ edge-updates on $\L(R^{+*})$. 
\end{lemma}
\begin{proof}
  \cref{lem:setESizeN} argues that there are only $O(n)$ edges that
  ever appear on $\L(R^{+*})$. Let $E$ be the collection of these edges.
  What remains to show is that each
  edge in $E$ is inserted into $\L(R^{+*})$ only once. Again, we partition the lines $R^*$ into pink lines $P^*$ with positive slope and green lines $G^*$ with negative slope.

\begin{figure}[tb]
    \centering
    \includegraphics{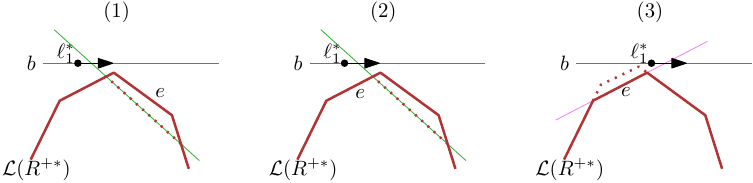}
    \caption{The three cases in which an edge $e$ can be deleted from $\L(R^{+*})$.}
    \label{fig:wedgeBlueShiftingUpdates}
\end{figure}

An edge $e \in E$ can be deleted from $\L(R^{+*})$ in only three
cases. For each case we argue that it can never be reinserted (see
Figure \ref{fig:wedgeBlueShiftingUpdates}). Edge $e$ is deleted when:

\begin{description}
\item[$\ell_1^*$ crosses a green line $g$ that lies below $e$.] Line
  $g$ will remain in $G^{+*}$, and thus $e$ will never appear on
  $\L(R^{+*})$ again.
\item[$\ell_1^*$ crosses a green line $g$ that intersects $e$.] We
  view this as deleting $e$, and inserting a new shorter edge. As in
  the previous case, line $g$ will always stay in $G^{+*}$, and thus
  $e$ will never appear on $\L(R^{+*})$ as a whole again.
\item[$\ell_1^*$ crosses a pink line $p$ containing $e$.] Line $p$ is
  permanently removed from $P^{+*}$, so clearly edge $e$ will never be
  inserted again.
\end{description}

So every edge $e \in E$ is inserted into $\L(R^{+*})$ once, and
therefore deleted at most once. Since $E$ has size $O(n)$,
it follows that there are only $O(n)$ edge-updates.
\end{proof}


\begin{lemma}
\label{lem:explicitEnvelope}
While shifting $\ell_1^*$ along $b$ we can maintain an explicit tree of $\L(R^{+*})$ in $O(n \log n)$ total time.
\end{lemma}
\begin{proof}
  We maintain the lower envelope of the set $R^{+*}$ in the data
  structure of Brodal and Jacob~\cite{brodal2002dynamic}, which can
  handle inserting and deleting a line in $O(\log n)$
  time. Furthermore, given a line $r$ it can report the edges of
  $\L(R^{+*})$ intersected by $r$, if any, and it can report the
  neighbouring edges of a given edge of $\L(R^{+*})$ in
  $O(\log n)$ time.

  We also need to support insertions and deletions into the explicit
  tree of $\L(R^{+*})$. We show that we can enumerate all edges that
  are to be inserted or deleted into this structure at an update to
  $R^{*+}$ in $O(\log n)$ time each. We can then delete and insert
  these edges from our explicit tree in $O(\log n)$ time each again
  (these are standard balanced binary tree operations), proving the
  lemma.

\begin{description}
\item[Insertion of a line $r$ into $R^{+*}$.] We query the Brodal and
  Jacob data structure for the edges of $\L(R^{+*})$ intersected by
  $r$ in $O(\log n)$ time. If $r$ and $\L(R^{+*})$ do not intersect
  we are done, so assume they do intersect. By general position, we
  can assume there are two intersections. All edges between these
  intersection are to be deleted. We walk through them and report
  them in $O(\log n)$ time each using the neighbour-queries of the
  Brodal and Jacob data structure. We insert exactly three new edges:
  one new edge on $r$, and two shortened intersected edges. Finally,
  we update the Brodal and Jacob data structure by inserting $r$ in
  $O(\log n)$ time.

\item[Deletion of a line $r$ from $R^{+*}$.] We do almost the same as for an insertion, except we change the order. We first update the Brodal and Jacob data structure, and only afterwards query for the intersections of $\L(R^{+*} \setminus \{ r \})$ and $r$. All edges between these intersections are to be inserted, and can be reported as above. Again there are three edges to be deleted: one on $r$ and two intersected by $r$.
\qedhere
\end{description}
\end{proof}


With the explicit tree at hand, we require only $O(n)$ queries to the oracle during the entire shifting process to maintain an optimal point $\ell_2^*$.

\begin{lemma}
  \label{lem:maintainMax}
  We can maintain an edge of $\L(R^{+*})$ with maximum level
  w.r.t. $B^{+*}$ while shifting $\ell_1^*$ along $b$ using $O(n)$
  queries to the oracle and $O(n\log n)$ additional
  time.
\end{lemma}
\begin{proof}
We maintain an explicit tree of $\L(R^{+*})$ using Lemma \ref{lem:explicitEnvelope}, and augment it. For a node $u$, let $\max(u)$ be the highest level of any edge in the
subtree of $u$. We wish to maintain $\max(u)$ for all nodes, but doing
so explicitly is expensive: when we cross a blue line we may need to
increment $\max(u)$ for many nodes. Therefore, we do so implicitly. We
augment each node $u$ with a \emph{level buffer} $b(u)$ and a
\emph{partial max} $\max'(u)$, under the invariant that
$\max(u) = \max'(u) + \sum_{a \in A(u)} b(a)$, where $A(u)$ is the set
of ancestors of $u$ (including $u$ itself). In particular, for the root node $u$ we have that $\max'(u)$ is the maximum level of any edge of $\L(R^{+*})$ w.r.t. $B^{+*}$, which is exactly what we want to maintain.

Whenever we visit a node $u$, we \emph{propagate the buffer}: we apply the buffer to the node by setting $\max'(u) \pluseq b(u)$, and for each child $c$ of $u$ we propagate the buffer by setting $b(c) \pluseq b(u)$. Lastly, we reset the buffer $b(u) = 0$. If the invariant held before this, it will still hold after this. This procedure ensures that for any node $u$ we visit it holds that $\max'(u) = \max(u)$: while walking towards $u$ we propagated the buffer at every step, so $b(a) = 0$ for all ancestors $a \in A(u)$.

There are two types of events: $\ell_1^*$ crosses a blue line, or $\ell_1^*$ crosses a red line.

\begin{figure}
    \centering
    \includegraphics{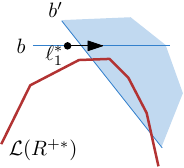}
    \caption{When we cross a blue line, the level of all edges above it changes.}
    \label{fig:wedgeBlueLevelDatastructure}
\end{figure}

\begin{description}
\item[$\ell_1^*$ crosses a blue line $b'$.] Assume $b'$ has negative
  slope, so we insert $b'$ into $B^{+*}$; the other case is
  analogous. See \cref{fig:wedgeBlueLevelDatastructure}. The levels of all edges of
  $\L(R^{+*})$ fully above $b'$ increase by one, since there is one more blue line below them. We can
  find $O(\log n)$ nodes of the tree representing these edges, and for
  all these nodes $u$ we increment $b(u)$ by one, satisfying the
  invariant for all edges in these subtrees. For the (at most) two
  edges intersected by $b$ the level may or may not have changed. So we query the oracle for the level $\level_{B^{+*}}(e)$ of
  such an edge $e$ w.r.t. $B^{+*}$. We then walk towards the leaf node
  $u$ containing $e$ (propagating the buffers as we go), and set
  $\max'(u) = \level_{B^{+*}}(e)$.

\item[$\ell_1^*$ crosses a red line $r$.] We update the explicit tree as normal, with one addition: when we create a
  leaf node $u$ for a new edge $e$ we query the oracle for the level
  $\level_{B^{+*}}(e)$ of $e$ w.r.t. $B^{+*}$, and set
  $\max'(u) = \level_{B^{+*}}(e)$.
\end{description}


For both events we spend only $O(\log n)$ time per edge-update, so maintaining the tree still takes $O(n \log n)$ total time. Additionally, we query the oracle $O(n)$ times: once per edge insertion,
and at most twice per blue line crossed. 
\end{proof}


\subparagraph{Collecting queries to the oracle.} What remains is to describe how
to implement the oracle that answers the queries. However, observe
that the set of queries to the oracle is fixed.
Also, the answer to an
oracle query is independent of the answers to earlier queries, and the answers of queries do not influence which future queries will be
performed. Therefore, we can perform a 'dry'-run of the algorithm
where we collect all queries, and then answer them in bulk
(during this 'dry'-run there is no need to maintain the values $b(u)$
and $\max'(u)$ yet). As we will see, this allows us to answer these
queries efficiently.

Since each blue line generates $O(n)$ queries
(Lemma~\ref{lem:maintainMax}), we have a total of $O(n^2)$
queries. Once we have the answers to all these queries, We once again
run the algorithm for each blue line $b$. In this 'real'-run of the
algorithm we can answer queries in $O(1)$ time, and thus compute
an optimal pair of points $\ell_1^*$ and $\ell_2^*$ with $\ell_1^*$ on
$b$ in $O(n\log n)$ time (Lemma~\ref{lem:maintainMax}). This then
leads the following result.

\begin{lemma}
  \label{lem:with_query_oracle}
  Given two sets of $n$ points $B,R \subset \R^2$, we can construct a
  wedge containing as many points of $B$ as possible and no points of
  $R$ in $O(T(n) + n^2\log n)$ time, where $T(n)$ is the total time
  required to answer $O(n^2)$ oracle queries.
\end{lemma}

\subparagraph{Implementing the oracle.}  We will now show how to
implement the oracle that can answer a (batch of) the following
queries efficiently. Given a query $(\ell_1^*, s)$ consisting of a
point $\ell_1^*$ and a red segment $s$, we wish to find the maximum
level of any point on $s$ w.r.t. the set $B^{+*}$ of blue lines below
$\ell_1^*$. Maintaining the set $B^{+*}$ and answering queries fully
dynamically is difficult, so we will instead answer them in bulk.

We consider a red line $r$ and assume w.l.o.g. that it is horizontal. Let $Q$ be the set of queries whose line segment $s$ lies on
$r$, let $q_r = |Q|$, and let $P$ be the set of query points
$\ell_1^*$ corresponding to the queries in $Q$. See Figure \ref{fig:wedgeBlueOutliersSpanningTree}.

We pick an arbitrary query $(\ell_1^*, s) \in Q$. Let $b \in B^{+*}$ be a blue line with negative slope; the other case is analogous. Consider the intersection point $i$ between $b$ and $r$. For points $p \in s$ left of $i$, $b$ lies above $p$ and thus $b$ does not add to the level of $p$. For points $p \in s$ right of $i$, $b$ lies below $i$ and thus $b$ does add to the level of $p$. Consider all intersections between $r$ and lines in $B^{+*}$. We build a balanced binary tree on these intersections, ordered by $x$-coordinate, augmented such that each node also stores the point with the highest level in its subtree. Recall that $s$ is a line segment on $r$, and thus represents an $x$-interval on $r$. We can easily answer the query $(s,\ell_1^*)$ by finding the $O(\log n)$ nodes representing that interval and returning the maximum level of any point inside their subtrees.

To answer the other queries we can shift the point $\ell_1^*$ through
the arrangement $\A(B^*)$. At every step, one blue line is inserted into or deleted from $B^{+*}$. We can update the binary tree in $O(\log n)$ time per step, meaning that if we reach a point $p \in P$ we can answer the corresponding query in $O(\log n)$ time. So, if we can walk through $\A(B^*)$ in $S$ steps, visiting all points $P$, we can answer all queries $Q$ in $O(S \log n)$ time. For this we use a \emph{spanning tree}.

A spanning tree on $P$ is a set of $q_r-1$ edges connecting the points
in $P$. The
stabbing number of a spanning tree is the maximum number of edges of
the tree that can be intersected by a single line. With high
probability (with probability $1 - 1/q^c$ for some arbitrarily large
constant $c$) we can build a spanning tree $T$ on $P$ with stabbing
number $O(\sqrt{q})$ in $O(q \log q)$
time~\cite{chan12partition_tree}. Thus, with high probability each blue line intersects $T$ at most $O(\sqrt{q_r})$ times, and thus there are $O(n \sqrt{q_r})$ intersections between $T$ and $B$. There is randomness involved here; if there are more then $O(n \sqrt{q_r})$ intersections between $T$ and $B$ we simply rebuild the tree, making our final running time expected rather than deterministic.

If we follow the spanning tree while walking through $\A(B^*)$ we will thus cross $O(n\sqrt{q_r})$ blue lines in total while visiting all points $P$, meaning we can answer all queries $Q$ on $r$ in $O(n \sqrt{q_r} \log n)$ time (note that the $O(q_r \log q_r)$ time to build the spanning tree is dominated by $O(n \sqrt{q_r} \log n)$ for any $q_r =O(n^2)$). Doing this for all
lines $r \in R$ takes $O(\sum_r (n \sqrt{q_r}\log n))$
time. Recall we have $O(n^2)$ queries in total, so we have
$\sum_r q_r = O(n^2)$. 
Since $\sqrt{\cdot}$ is a concave function, we have $\sqrt{a} + \sqrt{b} \leq \sqrt{2(a+b)}$ for any non-negative values $a$ and $b$.
More generally, $\sum_i \sqrt{x_i} \leq \sqrt{n\left( \sum_i x_i \right)}$ for any $n$ non-negative values $x_i$.
In particular, $\sum_r \sqrt{q_r} \leq \sqrt{n\left( \sum_r q_r \right)} = O(n^{3/2})$.
Therefore, we have an $O(\sum_r (n \sqrt{q_r}\log n)) = O(n^{5/2}\log n)$ time
algorithm to answer all queries over all red lines.

\begin{lemma}
  \label{lem:batch_query_lemma}
We can answer $O(n^2)$ queries in expected $O(n^{5/2} \log n)$ time.
\end{lemma}

Together with \cref{lem:with_query_oracle} this yields an expected $O(n^{5/2} \log n + n^2 \log n) = O(n^{5/2} \log n)$ time algorithm to find an optimal \North wedge.

We can symmetrically find an optimal \South wedge by assuming the wedge lies below both $\ell_1$ and $\ell_2$. Similarly by assuming the wedge lies below $\ell_1$ and above $\ell_2$ we can find an optimal \West or \East wedge.

\begin{theorem}
Given two sets of $n$ points $B,R \subset \R^2$, we can construct a wedge containing as many points of $B$ as possible and no points of $R$ in expected $O(n^{5/2} \log n)$ time.
\end{theorem}

\subsection{Wedge separation with both outliers}
\label{sub:single_wedge_both_outliers}
We now consider the case where we allow and minimize both red and blue
outliers. We show how to find an optimal \South wedge; finding an
optimal \West, \North, or \East wedge can be done analogously. We first
study the decision version of this problem: given an integer $k'$, does
there exist a \South wedge $\W_B$ with at most $k'$ outliers? We
present an $O(nk'^2 \log^3 n)$ time algorithm to solve this decision problem.
Using exponential search to guess the optimal value $k$ (i.e., guessing $k' = 1,2,4,8 \dots$ and then binary searching in the remaining interval) then leads
to an $O(nk^2 \log^3 n \log k)$ time algorithm to compute a wedge $\W_B$ that minimizes $k$.

\begin{figure}
\centering
\begin{minipage}{.5\textwidth}
  \centering
    \includegraphics{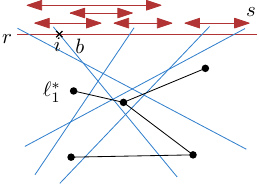}
    \caption{A set of queries on a red line $r$, with a spanning tree on $P$. Line $b$ contributes to the level of all points right of $i$.}
    \label{fig:wedgeBlueOutliersSpanningTree}
\end{minipage}
\hspace{.05\textwidth}
\begin{minipage}{.43\textwidth}
  \centering
    \includegraphics{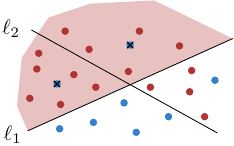}
    \caption{All points above $\ell_1$ lie outside the \South wedge. After fixing $\ell_1$, we are left with a halfplane separation problem.}
    \label{fig:wedgeBothOutliers}
\end{minipage} 
\end{figure}

In~\cref{lem:nkCandidateLines} we construct a small candidate set of lines that contains a line $\ell_1$ that is used by an optimal wedge.
Then our algorithm considers each line in this set and constructs an optimal wedge using it.

\begin{lemma}
\label{lem:nkCandidateLines}
In $O(nk' \log n)$ time, we can construct a set of $O(nk')$ lines that contains a line $\ell_1$ used by an optimal wedge.
\end{lemma}
\begin{proof}

Consider any line $\ell$ and suppose it is used in a \South wedge as the line $\ell_1$. Let $B^+$ and $B^-$ be the set of blue
points above, respectively below, $\ell_1$. Since we are looking for a \South wedge, all $k_{1} =
|B^+|$ blue points above $\ell_1$ are misclassified, regardless of
line $\ell_2$ (see~\cref{fig:wedgeBothOutliers}). Therefore, any
line with $k_{1} > k'$ blue points above it is not a suitable candidate
for $\ell_1$. In the dual plane, this means $\ell_1^*$ must lie in the
\emph{$(\leq k')$-level $L_{\leq k'}(B^*)$ of $B^*$}, the set of points with at most $k'$ lines of $B^*$ below them.
With a slight abuse of notation, we use $L_{\leq k'}(B^*)$ to refer to the sub-arrangement of $\A(B^*)$ that lies in the $(\leq k')$-level. The complexity of $L_{\leq k'}(B^*)$ is $O(nk')$, and we can construct $L_{\leq
  k'}(B^*)$ in $O(nk' + n \log n)$
time~\cite{everett1993lessThanKLevels}. 

Consider any line $r \in R^*$, and observe that it intersects $L_{\leq k'}(B^*)$ at most $O(k')$ times.
This follows from the fact that we can decompose $L_{\leq k'}(B^*)$ into $O(k')$ concave chains~\cite{chan05low_dimen_linear_progr_violat}, and that $r$ intersects each such chain at most twice in general position.
We can thus explicitly compute all the $O(nk')$ red-blue intersections in $L_{\leq k'}(B^*)$ in $O(nk' \log n)$ time, and by~\cref{lem:redBlueOptimum} these red-blue intersections contain the dual of a line used in an optimal wedge.
\end{proof}

Fix $\ell_1$ to be any candidate line from~\cref{lem:nkCandidateLines}. We wish to
find another line $\ell_2$ such that the wedge $\South(\ell_1,
\ell_2)$ misclassifies at most $k'$ blue points. Since all points above $\ell_1$ are outside the wedge regardless of the line $\ell_2$, we need to consider only the points $B^-$ and $R^-$
below $\ell_1$. Recall that the choice of $\ell_1$ already misclassifies $k_1$ blue points. Thus, we wish to find a line $\ell_2$ that misclassifies at most $k_2 = k' - k_1$ points from $B^-$ and $R^-$. 
That is, a line $\ell_2$ such that the number of points from $R^-$ below it plus the number of points from $B^-$ above it is at most $k_2$. This is exactly the halfplane separation problem with both
outliers, which we can solve using Chan's algorithm in $O((n + (k_2)^2)\log n)$ time~\cite{chan05low_dimen_linear_progr_violat}. Doing this for all $O(nk')$ candidate lines
results in an $O(nk' (n + k^2) \log n) = O((n^2 k' + n k'^3) \log n)$ time
algorithm. Below we improve on this, by avoiding to recompute $\ell_2$ from scratch every time.

\subparagraph{Solving the halfplane separation problem dynamically.}
Consider again the set of candidate lines of~\cref{lem:nkCandidateLines}, and in particular their dual points. By walking through the arrangement $L_{\leq k'}(B^*)$, we can visit all $O(nk')$ candidate points $\ell_1^*$ in $O(nk')$ steps, such that at each step we cross only one (red or blue) line. This means that only a single point is inserted in or deleted from the sets $B^-$ and $R^-$ per step. Rather than computing $\ell_2$ from scratch after every step now, we maintain it dynamically.

We build the data structure of~\cite{glazenburg2024dynamicsvm} that, given a value $k'$, maintains a line $\ell_2$ that misclassifies as few points from $R^-$ and $B^-$ as possible under insertions and deletions of red and blue points; if no line misclassifying at most $k'$ points exists, the data structure reports this.
The updates for the data structure have to be given in a 'semi-online' manner, which means that whenever a point is inserted we have to know when it is going to be deleted. This is not a problem in our case, since we can precompute all insertions and deletions on $R^-$ and $B^-$ by completing the walk through $L_{\leq k'}(B^*)$ before actually updating the data structure.

The data structure reports an optimal line $\ell_2$ that misclassifies $k_2 \leq k'$ points of $R^-$ and $B^-$, so the wedge $\South(\ell_1,\ell_2)$ then misclassifies $k_1 + k_2$ points. If $k_1 + k_2 \leq k'$ then we have found a wedge misclassifying at most $k'$ points, so we are done with the decision problem, otherwise we move on to the next candidate for $\ell_1$. 
If the data structure reports that there exists no line $\ell_2$ misclassifying at most $k'$ points, then we also move on to the next candidate.

The data structure has update time $O(k' \log^3 n)$ per insertion and deletion, and there are $O(nk')$ updates. Therefore, given a value $k'$, we can find a \South wedge with at most $k'$ outliers, if it exists, in $O(nk'^2 \log^3 n)$ time. Using exponential search for the optimal value $k$ then gives an optimal \South wedge in $O(nk^2 \log^3 n \log k)$ time. As in the previous section, we can similarly find \North, \West, and \East wedges.

\begin{theorem}
  \label{thm:single_wedge_both_outliers}
    Given two sets of $n$ points $B, R \subset \R^2$, we can construct a wedge $\W_B$ minimizing the total number of outliers $k$ in $O(nk^2 \log^3 n \log k)$ time.
\end{theorem}

\subparagraph{An output sensitive algorithm for minimizing $k_B$.} We
can use a very similar technique in the case where we allow only blue
outliers. We again study the decision problem for a given value $k_B'$,
where we wish to find a \South wedge containing no red points and all
and at most $k_B'$ blue points.  We then again use exponential search
to find the optimal value $k_B$.

We can use the same approach as above, i.e., constructing $L_{\leq k_B'}(B^*)$ and walking through it to find all candidate lines for $\ell_1$. At each step, we now wish to find a line
$\ell_2$ that lies below all red points $R^-$ and above as many blue
points $B^-$ as possible. This is the halfplane separation problem
with blue outliers, which we solved in $O(n \log n)$ time in~\cref{sec:Single_line_separation_one-sided}. Applying this algorithm to all $O(nk_B')$
candidates for $\ell_1$ results in an $O(n^2 k_B' \log n)$ time
algorithm for the decision problem, or an $O(n^2 k_B \log n \log k_B)$ time algorithm for minimizing $k_B$. Again we would like to dynamically maintain $\ell_2$, rather than recomputing it from scratch every step. The halfplane separation data structure from before~\cite{glazenburg2024dynamicsvm} cannot be used as is, since it allows both blue and red outliers. We show below how to
adapt (or rather, simplify) the data structure to allow and minimize only blue
outliers.

We work in the dual, where we wish to maintain a point $\ell_2^*$ that lies above all red lines $R^{-*}$ and below as many blue lines $B^{-*}$ as possible. The data structure maintains a concave chain decomposition of
$L_{\leq k_B'}(B^{-*})$, consisting of $O(k_B' \log n)$ chains covering
$L_{\leq k_B'}(B^{-*})$. Additionally, it maintains a convex chain decomposition of
$L_{\geq |R^{-*}|-k_B'-1}(R^{-*})$ (the $\leq k_B'$ level of $R^{-*}$ as seen from above). By~\cref{lem:redBlueOptimum} the optimum is one of the
$O(k_B'^2 \log^2 n)$ red-blue intersections between the two sets of chains.
Maintaining these intersections can be done in $O(k_B' \log^2 n)$ time. The original data structure uses a dynamic partition tree to maintain the best point among all these intersections.
However, this is not necessary in our case. Our optimum must lie on $\U(R^{-*})$ (the upper envelope of $R^{-*}$), so we can simply maintain $\U(R^{-*})$ rather than a convex chain
decomposition of $L_{\geq |R^{-*}|-k_B'-1}(R^{-*})$. This can be done in $O(\log^2 n)$ time using the data structure by Overmars and van Leeuwen~\cite{overmars1981dynamicCH}. This reduces the number of intersections between red and blue chains to only $O(k_B' \log n)$, since $L_{|R^{-*}| - 1}(R^{-*})$ intersects each concave blue chain at most twice. This means we can afford to iterate through all intersections after every update to find our
optimum, rather than using the dynamic partition tree as in the original
version. This reduces the update time from $O(k_B' \log^3 n)$ to $O(k_B' \log^2 n)$.

Iterating through all $O(nk_B')$ candidate lines for $\ell_1$ and dynamically maintaining an optimal $\ell_2$ at every step thus gives us an $O(nk_B'^2 \log^2 n)$ time algorithm to solve the decision problem for a given value $k_B'$. Using exponential search for the optimal value $k_B$ then gives us our result:

\begin{theorem}
  \label{thm:single_wedge_blue_outliers}
Given two sets of $n$ points $B, R \subset \R^2$, we can construct a wedge $\W_B$ minimizing the number of blue outliers $k_B$ in $O(nk_B^2\log^2 n \log k_B)$ time.
\end{theorem}

\section{Separation with a double wedge}
\label{sec:double_wedge}
In this section, we consider the case when $\W_B$ is a double wedge. When we do not allow blue outliers (i.e., when we are minimizing $k_R$), there are only $\Theta(n)$
relevant double bowtie wedges w.r.t. $B$ (pairs of antipodal faces of arrangement $B^*$). In contrast, there can be
$\Theta(n^2)$ relevant hourglass type double wedges (all faces of arrangement $B^*$). Hence, dealing
with hourglass wedges is harder: we give an $O(n^2)$ algorithm for finding a bowtie wedge $\W_B$ while minimizing $k_R$, but an $O(n^2 \log n)$ algorithm for finding an hourglass wedge $\W_B$ while minimizing $k_R$ (note that this second problem is equivalent to finding a bowtie wedge $\W_B$ while minimizing $k_B$, after swapping the colors of $R$ and $B$).

Bertschinger
\etal~\cite{bertschinger22inter_doubl_wedge_arran} observed similar
behavior when dealing with intersections of double wedges. In case we
have some lower bound on the interior angle $\alpha\pi \leq \pi$ of our double wedge $\W_B$, we refer to such double wedges as
\emph{$\alpha$-double wedges}, we can rotate the plane by $i\alpha\pi$,
for $i \in 0,..,O(1/\alpha)$, and run our (faster) bowtie algorithm each
of them. In at least one of these copies the optimum $\W_B$ is a bowtie
type wedge. This then leads to an $O(n^2/\alpha)$ time algorithm for
minimizing $k_R$. If we do not have any bound on $\alpha$, it is not clear how to find a rotation that turns $\W_B$ into a bowtie
type double wedge, and thus finding any double wedge $\W_B$ (hourglass or bowtie) while minimizing $k_R$ takes $O(n^2 \log n)$ time.

In Section~\ref{sub:Hourglass_wedge_separation_with_both_outliers} we describe an output-sensitive algorithm when both types of outliers are allowed, similar to \cref{sub:single_wedge_both_outliers}. 

\subsection{Bowtie wedge separation with red outliers}
\label{sub:bowtie_red}

We work in dual space, where a bowtie wedge containing all of $B$ dualizes to a line segment intersecting all lines of $B^*$.
Any line of $R^*$ that is also intersected corresponds to a red point in the double wedge, which is an outlier.
Hence we focus on computing a segment that intersects all lines of $B^*$ and as few of $R^*$ as possible.

Observe that the only segments intersecting all lines of $B^*$ have their endpoints in antipodal outer faces of $\A(B^*)$. We can construct the outer faces in $O(n \log n)$ time~\cite{bringmann22computing_zone}, since the outer faces are the zone of the boundary of a sufficiently large rectangle (one that contains all vertices of the arrangement).
With the outer faces constructed, we can apply a very similar algorithm to the one in Section \ref{sub:Single_wedge_with_red_outliers} on each pair of antipodal faces (where for the parameter space, lines of type $c$ and $d$ add two forbidden quadrants rather than one).
This gives an $O(n^2 \log n)$ time algorithm for bowtie double wedge classification with red outliers.

Considering the running time is super-quadratic, we opt to construct the entire arrangement $\A(B^* \cup R^*)$ of all lines explicitly.
This takes $O(n^2)$ time (see e.g.~\cite{deberg08computational_geometry}), and as we show next, allows us to shave off a logarithmic factor.

Let $P, Q$ be the boundary chains of a pair of antipodal outer faces of $\A(B^*)$, made up of a total of $m$ edges.
We assume for ease of exposition that $P$ and $Q$ are separated by the $x$-axis, with $P$ above and $Q$ below the axis.
We distinguish between two types of red lines: \emph{splitting lines} and \emph{stabbing lines}.
Splitting lines intersect both $P$ and $Q$, while stabbing lines intersect at most one of $P$ and $Q$. Note that a line is a splitting line for exactly one pair of antipodal faces, but can be a stabbing line for multiple pairs of antipodal faces. For two points $p$ and $q$, let $\stab(p,q)$ be the number of stabbing lines that $\ol{pq}$ intersects, and similarly define $\splt(p,q)$.
Let $s$ be the number of splitting lines for the pair of faces $P,Q$.

\begin{lemma}
    We can construct a line segment with endpoints on $P$ and $Q$ that intersects as few red lines as possible in $O(s^2 + m + n)$ time.
\end{lemma}
\begin{proof}
    See Figure \ref{fig:doubleWedgeRedOutliers} for an illustration. The $s$ splitting lines partition $P$ and $Q$ into $s+1$ chains each.
    Let $P_0, \dots, P_{s}$ be the chains partitioning $P$ and let $Q_0, \dots, Q_{s}$ be the chains partitioning $Q$, both in clockwise order along $P$ and $Q$.
    Consider some pair of chains $P_i, Q_j$. Note that all segments starting in $P_i$ and ending in $Q_j$ intersect the same number of splitting lines, i.e., $\forall p_1, p_2 \in P_i, \forall q_1, q_2 \in Q_j: \splt(p_1,q_1) = \splt(p_2,q_2) = \splt(P_i,Q_j)$. Therefore the best segment from $P_i$ to $Q_j$ is the one that intersects the fewest stabbing lines. For points $p \in P_i, q \in Q_j$, let $x_p$ be the number of stabbing lines above $p$, and $y_q$ the number of stabbing lines below $q$, and note that $\stab(p,q) = n - s - x_p - y_q$. Thus, the best segment from $P_i$ to $Q_j$ is $\ol{p_iq_j}$, where $p_i = \argmax_{p \in P_i} x_p$, and $q_j = \argmax_{q \in Q_j} y_q$. Note that $p_i$ does not depend on $Q_j$, and vice versa.
    
    We compute these points $p_i$ for all chains $P_i$ (and symmetrically $q_j$ for all chains $Q_j$) as follows. We move a point $p$ clockwise along $P$ in the arrangement $\A(B^* \cup R^*)$, maintaining $x_p$, as well as $p_{\max}$ and $x_{p_{\max}}$, the (point attaining the) maximum value of $x_p$ encountered so far on the current chain $P_i$. When we cross a stabbing line $\ell$ we increment or decrement $x_p$, depending on the slope of $\ell$.
    Specifically, if the slope of $\ell$ is greater than the slope of the edge of $P$ we are currently on then we decrement $x_p$, otherwise we increment it. When we reach the end of chain $P_{i}$, i.e., when we cross a splitting line or reach the end of $P$, we set $p_i = p_{\max}$, and reset $p_{\max}$ and $x_{p_{\max}}$.
    This procedure takes $O(m + n)$ time.
    
    For each segment $\ol{p_i q_j}$, we now know $\stab(p_i, q_j) = n - s - x_{p_i} - y_{q_j}$.
    Next we show that we can compute the number of splitting lines $\splt(P_i, Q_j)$ intersected by segments with endpoints on $P_i$ and $Q_j$, for all pairs, in $O(s^2)$ time.
    For this we use dynamic programming.
    Let $\ell_i$ be the splitting line in between chain $P_i$ and $P_{i+1}$.
    We compute the number of splitting lines $\splt(P_i, Q_i)$ between each pair of chains with the following recurrence (recall that the partitionings of $P$ and $Q$ are in clockwise order):
    \begin{align*}
        \splt(P_i, Q_j) &=
        \begin{cases}
            s - j                       & \text{if $i = 0$,}\\
            \splt(P_{i-1}, Q_{j}) + 1   & \text{if segment $\ol{p_{i} q_j}$ intersects $\ell_{i-1}$,}\\
            \splt(P_{i-1}, Q_{j}) - 1   & \text{if segment $\ol{p_{i} q_j}$ does not intersect $\ell_{i-1}$.}
        \end{cases}
    \end{align*}
    We compute the $O(s^2)$ values for $\splt(P_i, Q_j)$ in $O(s^2)$ time with dynamic programming.
    Having computed both $\stab(p_i, q_j)$ and $\splt(P_i, Q_j)$ for all pairs of chains, we compute the pair minimizing $\stab(p_i, q_j) + \splt(P_i, Q_j)$ in $O(s^2)$ additional time by iterating through them.
    The segment connecting this pair intersects the fewest red lines.
\end{proof}

\begin{figure}[tb]
    \centering
    \includegraphics{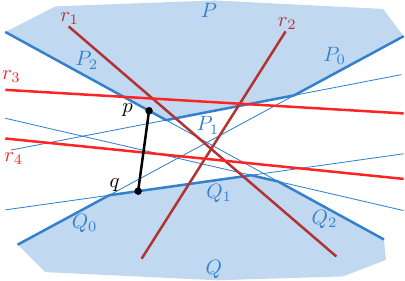}
    \caption{Two antipodal faces $P$ and $Q$, with two splitting lines $r_1, r_2$ and two stabbing lines $r_3,r_4$, and an optimal segment $\ol{pq}$ from $P$ to $Q$.}
    \label{fig:doubleWedgeRedOutliers}
\end{figure}

There are $O(n)$ pairs of antipodal blue faces $P$ and $Q$.
For the $x^{\mathrm{th}}$ pair, let $m_x$ be their total complexity and $s_x$ be the number of splitting lines.
We apply the above algorithm to each pair, leading to total time $O(\sum_x (s_x^2 + m_x + n))$.
The total complexity of all outer faces is $O(n)$, and a red line is a splitting line for exactly one pair of antipodal faces.
Hence the total running time simplifies to $O(n^2)$.

\begin{lemma}
\label{lem:bowtie_wedge_red_outliers}
    Given two sets of $n$ points $B, R \subset \R^2$, we can construct the bowtie double wedge $\W_B$ minimizing the number of red outliers $k_R$ in $O(n^2)$ time.
\end{lemma}

\subsection{Bowtie wedge separation with blue outliers}
\label{sub:bowtie_blue}

Again we work in dual space, where a bowtie double wedge $\W_B$ containing some of $B$ and none of $R$ dualizes to a line segment intersecting some lines of $B^*$ and none of $R^*$.
Any line of $B^*$ that is not intersected corresponds to a blue point not in $\W_B$, which is an outlier.
Hence we focus on computing a segment that intersects the most lines of $B^*$, while intersecting none of $R^*$.

Observe that the only segments intersecting no line of $R^*$ lie completely inside a face of $\A(R^*)$.
We construct the arrangements $\A(R^*)$ and $\A(B^* \cup R^*)$ in $O(n^2)$ time~\cite{deberg08computational_geometry}.

Consider a face $F$ of $\A(R^*)$.
We wish to compute the segment in $F$ that intersects the most blue lines, and which hence has the fewest blue outliers of any segment in $F$.
W.l.o.g. we only consider segments with endpoints on the boundary $P$ of $F$, since we can always extend a segment without introducing blue misclassifications.

Let $B^*_P$ be the set of blue lines intersecting $P$, which we report by scanning over $P$ inside the arrangement $\A(B^* \cup R^*)$.
This takes $O(|P| + |B^*_P|)$ time.
We reuse the parameter space tool from Section \ref{sub:Single_wedge_with_red_outliers}.
Fix an arbitrary point $o$ on $P$ and parameterize over $P$ in clockwise order, with $P(0) = P(1) = o$.
For a given blue line $b$ intersecting $P$ in points $P(b_1), P(b_2)$ with $b_1 < b_2$, a segment $\ol{P(\ell_1^*) P(\ell_2^*)}$ intersects $b$ if and only if $\ell_1^* \in [b_1, b_2]$ and $\ell_2^* \in [0, b_1] \cup [b_2, 1]$, or $\ell_1^* \in [0, b_1] \cup [b_2, 1]$ and $\ell_2^* \in [b_1, b_2]$.
This results in four blue intersection-regions in the parameter space, as in Figure \ref{fig:parameter_space_single_face}.

\begin{figure}[tb]
    \centering
    \includegraphics{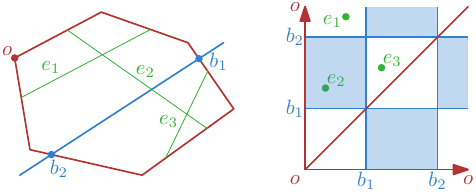}
    \caption{A face $F$ with its parameter space and the forbidden regions induced by the blue line.
    }
    \label{fig:parameter_space_single_face}
\end{figure}

We compute the intersection values $x_1, x_2$ for all lines in $B^*_P$ by scanning over $P$ in $\A(B^* \cup R^*)$, as we did for reporting all lines in $B^*_P$.
The segment $\ol{P(\ell_1^*) P(\ell_2^*)}$ intersecting the most blue lines corresponds to the point $(\ell_1^*, \ell_2^*)$ in the parameter space with maximum ply in the blue rectangles.
Similar to Lemma \ref{lem:minimum_ply}, where we compute the minimum ply point in a set of rectangles, we can compute the maximum ply point in this set of $|B^*_P|$ rectangles in $O(|B^*_P| \log |B^*_P|)$ time.

The total complexity of the sets $B^*_P$, over all faces of $\A(R^*)$, is at most the complexity of $\A(B^* \cup R^*)$, which is $O(n^2)$.
We therefore obtain a total running time of $O(n^2 \log n)$.

\begin{lemma}
\label{lem:bowtie_wedge_blue_outliers}
    Given two sets of $n$ points $B, R \subset \R^2$, we can construct a bowtie double wedge $\W_B$ minimizing the number of blue outliers $k_B$ in $O(n^2 \log n)$ time.
\end{lemma}

By combining~\cref{lem:bowtie_wedge_red_outliers,lem:bowtie_wedge_blue_outliers} we obtain the following result regarding double wedge classification with one-sided outliers:

\begin{theorem}
\label{thm:double_wedge_one-sided}
    Given two sets of $n$ points $B, R \subset \R^2$, we can construct two double wedges $\W_B$ and $\W_R$ that correctly classify $B$ and $R$, respectively, while minimizing the number of outliers of the other color, in $O(n^2 \log n)$ time.
\end{theorem}

\subsection{Double wedge separation with both outliers}
\label{sub:Hourglass_wedge_separation_with_both_outliers}
Consider now the problem of finding a double wedge that minimizes the
total number of outliers $k$. We show how to find an optimal hourglass
wedge, and by recolouring we can also find an
optimal bowtie wedge. We again consider the decision version of this problem: given a value $k'$, does there exist an hourglass wedge with at most $k'$ outliers? We then use exponential search to find the optimal value $k$. We use an approach similar to
Section~\ref{sub:single_wedge_both_outliers}, by considering a set of candidate options for the first line
$\ell_1$ and computing an optimal corresponding $\ell_2$ for each. Since the
choice of $\ell_1$ by itself does not misclassify any points yet, we have to consider all $O(n^2)$ red-blue intersections of the
arrangement of $B^*$ and $R^*$. 

\begin{figure}[tb]
    \centering
    \includegraphics{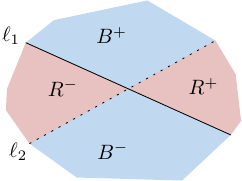}
    \caption{If we want $B$ to lie in the \North and \South wedges, then $B^+$ and $R^-$ should be above $\ell_2$, and $B^-$ and $R^+$ should be below $\ell_2$.}
    \label{fig:hourglassBothOutliers}
\end{figure}
 
For a fixed line $\ell_1$, again let $B^+$ and $R^+$ be the points
above $\ell_1$, and $B^-$ and $R^-$ the points below $\ell_1$. Recall that we want $B$ to be inside the \North and \South wedges, so we
wish to find a line $\ell_2$ that maximizes the number of points of
$B^+$ and $R^-$ above $\ell_2$ plus the number of points of $B^-$ and
$R^+$ below $\ell_2$.
See~\cref{fig:hourglassBothOutliers}. Let $P = B^+ \cup R^-$ and $Q =
B^- \cup R^+$. For a fixed line $\ell_1$, the best line $\ell_2$ is then the line that separates $P$ and $Q$ with the fewest outliers, with $P$ above $\ell_2$ and $Q$ below $\ell_2$. (This can also be seen as flipping the colors of all points below $\ell_1$.)

We can use Chan's algorithm for halfplane
separation with both outliers to solve this problem in
$O((n + k'^2)\log n)$ time~\cite{chan05low_dimen_linear_progr_violat}. Doing so for all $O(n^2)$ candidate lines $\ell_1$ would give an 
$O(n^2 \cdot (n + k'^2)\log n) = O((n^3 + n^2 k'^2)\log n)$ time algorithm.
Again, we can improve this time by maintaining the sets $P$ and $Q$, and an optimal line $\ell_2$ separating them, dynamically. By walking through $\A(R^* \cup B^*)$ we can visit all $O(n^2)$ candidate lines for $\ell_1$ in $O(n^2)$ steps, such that at every step a single point is inserted in or deleted from $P$ or $Q$. The same data structure used
in~\cref{sub:single_wedge_both_outliers} for dynamic halfplane separation~\cite{glazenburg2024dynamicsvm} suffices to maintain $\ell_2$. This data structure supports updates in $O(k' \log^3 n)$ time, giving us an $O(n^2 k' \log^3 n)$ time algorithm to solve the decision problem, i.e., to find if there exists an hourglass wedge with at most $k'$ outliers. Using exponential search to find the optimal value $k$ then gives us our result:

\begin{theorem}
  \label{thm:double_wedge_both_outliers}
  Given two sets of $n$ points $B, R \subset \R^2$, we can construct a
  double wedge $\W_B$ minimizing the total number of outliers $k$ in
  $O(n^2 k \log^3 n \log k)$ time.
\end{theorem}

\section{Concluding Remarks}
\label{sec:Concluding_Remarks}

We presented efficient algorithms for robust bichromatic classification of $R \cup B$ with at most two lines. Our results depend on the shape of the region containing (most of the) blue points $B$, and whether we wish to minimize the number of red outliers, blue outliers, or both. See Table~\ref{tab:overview_results}. Many of our algorithms reduce to the problem of computing a point with minimum ply with respect to a set of regions. We can extend these algorithms to support weighted regions, and thus we may support classifying weighted points (minimizing the weight of the misclassified points). It is interesting to see if we can support other error measures as well. 

There are also still many open questions. The most prominent questions are wheter we can design faster algorithms for the algorithms minimizing the total number of outliers $k$, in particular for the wedge and double wedge case. For the strip case, the running time of our algorithm $O(n^2 \log n)$ matches the worst case running time for halfplanes ($O((n + k^2)\log n)$, which is $O(n^2\log n)$ when $k = O(n)$), but it would be interesting to see if we can also obtain algorithms sensitive to the number of outliers $k$. Furthermore, it would be interesting to establish lower bounds for the various problems. In particular, are our algorithms for computing a halfplane minimizing $k_R$ optimal, and in case of wedges (where the problem is asymmetric) is minimizing the number of blue outliers $k_B$ really more difficult then minimizing $k_R$?

\bibliography{bibliography}

\end{document}